%% file: paper_arXiv.tex
\documentclass[11pt,a4paper]{article}
\usepackage{vmargin}
\setmarginsrb{1.0in}{1.0in}{1.0in}{1.0in}{0mm}{0mm}{5mm}{5mm}

\usepackage{amsfonts,amsmath,amssymb,stmaryrd}
\usepackage{amsthm}
\usepackage{thmtools,thm-restate}
\usepackage{graphicx}
\usepackage{color}
\usepackage{complexity}
\usepackage[utf8]{inputenc}
\usepackage{enumitem}
\usepackage[numbers,sort&compress,longnamesfirst]{natbib}
\usepackage{todonotes}
\usepackage{xspace}
\usepackage{comment}
\usepackage[vlined,slide,linesnumbered,ruled]{algorithm2e}
\usepackage[colorlinks=false,linkbordercolor=black,pdfborderstyle={/S/U/W 0.1}]{hyperref}

\usepackage{wrapfig}
\usepackage{placeins}

\usepackage{tikz}
\usepackage{cleveref}

\usetikzlibrary{shapes,calc}

\newtheorem{theorem}{Theorem}
\newtheorem{lemma}{Lemma}
\newtheorem{corollary}{Corollary}

\newcommand{\metatree}{Meta Tree\xspace}
\newcommand{\candidateblock}{Candidate Block\xspace}
\newcommand{\candidateblocks}{Candidate Blocks\xspace}
\newcommand{\metagraph}{Meta Graph\xspace}
\newcommand{\bridgeblock}{Bridge Block\xspace}
\newcommand{\bridgeblocks}{Bridge Blocks\xspace}
\newcommand{\Cinc}{\mathcal{C}_{inc}}
\newcommand{\CI}{\mathcal{C}_\mathcal{I}}
\newcommand{\CU}{\mathcal{C}_\mathcal{U}}
\newcommand{\profit}{\textnormal{profit}}

\newcommand{\ap}{{v_{\aps}}}
\newcommand{\aps}{a}

\crefformat{footnote}{#2\footnotemark[#1]#3}

\title{Efficient~Best Response~Computation for Strategic Network Formation under Attack}
\author{Tobias Friedrich\thanks{Algorithm Engineering Group, Hasso Plattner Institute Potsdam, Germany}  
\and Sven Ihde\footnotemark[1] 
\and Christoph Keßler\footnotemark[1] 
\and Pascal Lenzner\footnotemark[1] \thanks{\texttt{pascal.lenzner@hpi.de}}
\and Stefan Neubert\footnotemark[1] 
\and David Schumann\footnotemark[1]
}

\date{~}
\begin{document}

\maketitle

\begin{abstract} 
\noindent Inspired by real world examples, e.g. the Internet, research\-ers have introduced an abundance of strategic games to study natural phenomena in networks. Unfortunately, almost all of these games have the conceptual drawback of being computationally intractable, i.e. computing a best response strategy or checking if an equilibrium is reached is \NP-hard. Thus, a main challenge in the field is to find tractable realistic network formation models.

We address this challenge by investigating a very recently introduced model by Goyal et al. [WINE'16] which focuses on robust networks in the presence of a strong adversary who attacks (and kills) nodes in the network and lets this attack spread virus-like to neighboring nodes and their neighbors. 
Our main result is to establish that this natural model is one of the few exceptions which are both realistic and computationally tractable. In particular, we answer an open question of Goyal et al. by providing an efficient algorithm for computing a best response strategy, which implies that deciding whether the game has reached a Nash equilibrium can be done efficiently as well. Our algorithm essentially solves the problem of computing a minimal connection to a network which maximizes the reachability while hedging against severe attacks on the network infrastructure and may thus be of independent interest.   
\end{abstract}

	\input{Introduction}

	\input{Model}

	\input{Algorithm}

    \input{EmpiricalResults}

	\input{Uniform-attacker}
	\input{Conclusion}

	\bibliographystyle{abbrv}
	\bibliography{mpss2016farXiv}
	
\end{document}

%% file: Introduction.tex
\section{Introduction}
Many of today's important networks, most prominently the Internet, are essentially the outcome of an unsupervised decentralized network formation process among many selfish entities~\cite{Pap01}.
In the case of the Internet these selfish entities are Autonomous Systems (AS) which interconnect via peering agreements and thereby create a connected network of networks.
Each AS can be understood as a selfish player who strategically chooses a subset of other ASs to directly connect with.
Each inter-AS-connection is costly and yields a benefit and a risk.
The benefit is a reliable direct link towards the other AS.
However, such a connection may be used by malicious software and thus harbors the risk of collateral damage if a neighboring AS is attacked.

The field of strategic network formation, started by the seminal works of Jackson \& Wolinsky~\cite{JW96}, Bala \& Goyal~\cite{BG00} and Fabrikant et al.~\cite{Fab03}, studies the global structure and properties of networks formed by individual players making decentralized local strategic choices.
In all considered models there are players trying to optimize their own benefit, while minimizing their individual cost.
It is far from obvious why a collection of individual selfish strategies eventually results in useful and reliable network topologies like the Internet.
Studying the properties of such models aims for revealing insights about properties of existing naturally grown networks and inspiring methods to improve them.


Required features of any Internet-like communication network are reachability and robustness.
Such networks have to ensure that even in case of cascading edge or node failures caused by technical defects or malicious attacks, e.g. DDoS-attacks or viruses, most participating nodes can still communicate.
This important focus on network robustness has long been neglected and is now a very recent endeavor in the strategic network formation community, see e.g. \cite{Kli11,MMO15,GJKKM15,CLMM16}.
We contribute to this endeavor by proving that the very recently introduced natural model by Goyal et al.~\cite{GJKKM15,GJKKM15arxiv} is one of the few exceptions of a tractable network formation model. In particular, we provide an efficient algorithm for computing a utility maximizing strategy for their elegant model, which can be used to efficiently decide whether a network is in Nash equilibrium.
Thus, our algorithm allows the model of Goyal et al. to be used to predict real world phenomena in large scale simulations and to analyze real world networks. Moreover, predicted structural properties of equilibria, which are obtained by a tractable model, are more plausible since it can be assumed that computationally bounded agents in the real world will eventually end up in an equilibrium state or at least close to it.

\subsection{Related Work}
We focus on the model for strategic network formation with attack and immunization recently proposed by Goyal et al.~\cite{GJKKM15,GJKKM15arxiv}.
This model essentially augments the well-known reachability model by Bala \& Goyal~\cite{BG00} with robustness considerations. 
In particular, different types of adversaries are introduced which attack (and destroy) a node of the network.
This attack then spreads virus-like to neighboring nodes and destroys them as well.
Besides deciding which links to form, players also decide whether they want to buy immunization against eventual attacks. The model is the first model which incorporates network formation and immunization decisions at the same time. 

On the one hand, the authors of~\cite{GJKKM15,GJKKM15arxiv} provide beautiful structural results for their model. For example, showing that equilibrium networks are much more diverse than in the non-robust version, proving that the amount of edge overbuilding due to robustness concerns is small and establishing that equilibrium networks generally achieve very high social welfare.
On the other hand, the authors raise the intriguing open problem of settling the complexity of computing a best response strategy 
in their model\footnote{This question was raised in \cite{GJKKM15arxiv} for the maximum carnage adversary and is replaced in \cite{GJKKM15} with a reference to our preprint of the present paper. The preprint also appeared as a Brief Announcement~\cite{FIKLNS}}.

Computing a best response in network formation games can be done in polynomial time for the non-robust reachability model~\cite{BG00} and if the allowed strategy changes are very simple~\cite{L12}.
However, these examples are exceptions.
The existence of an efficient best response algorithm for a network formation game is in general a rare gem. For almost all related network formation models, e.g. \cite{Fab03,MS10,Ehs15,Bilo14a,Bilo15,CL15,CLMM16}, where players strive for a central position in the network, it has been shown that the problem is indeed \NP-hard.

To the best of our knowledge, besides the model by Goyal et al.~\cite{GJKKM15,GJKKM15arxiv} there are only a few other models which combine selfish network formation with robustness considerations and all of them consider a much weaker adversary which can only destroy a single edge.
The earliest are models by Bala \& Goyal~\cite{BG03}
and Kliemann~\cite{Kli11}, both essentially augment the model by Bala \& Goyal~\cite{BG00} with single edge failures.
Other related models are by Meirom et al.~\cite{MMO15} and Chauhan et al.~\cite{CLMM16}.
Both latter models consider players who try to be as central as possible in the created networks but at the same time want to protect themselves against single edge failures.
In~\cite{MMO15} heterogeneous players
are considered whereas in~\cite{CLMM16} all players are homogeneous.
The complexity status for computing a best response was only settled for the model by Chauhan et al.~\cite{CLMM16} where this problem was proven to be \NP-hard.

Apart from network formation games, also vaccination games, e.g. \cite{ACY06,CDK10,KRSS10,SAV14}, are related. There the network is fixed and the selfish nodes only have to decide if the want to immunize or not. Computing a best response in these models is trivial (there are only two strategies) but pure Nash equilibria may not exist. \vspace*{0.1cm} 

\subsection{Our Contribution}
We establish that the natural model by Goyal et al.~\cite{GJKKM15,GJKKM15arxiv} is one of the few examples of a tractable realistic model for strategic network formation and thereby answer an open question by these authors. In particular, we provide an efficient algorithm for computing a best response strategy for their main model, i.e. the ``maximum carnage'' adversary which tries to kill as many nodes as possible, and for the natural variant which employs the even stronger random attack adversary. 
Moreover, we employ our algorithm in empirical simulations which augment the extensive simulations presented in~\cite{GJKKM15,GJKKM15arxiv}, where players were allowed to perform only heavily restricted strategy changes.
We observe fast and reliable convergence, despite the fact that the authors of~\cite{GJKKM15,GJKKM15arxiv} have shown that best response dynamics can cycle, a phenomenon also observed in other network formation games~\cite{KL13}.

\subsection{Organization of the Paper}
In Section~\ref{sec:model} we introduce the model by Goyal et al.~\cite{GJKKM15,GJKKM15arxiv} and some additional definitions and notation.
Our main contribution follows in Section~\ref{sec:algorithm}, where we introduce the main ideas and a detailed description of our algorithm. Several subroutines are used and we first show how these subroutines work together to solve the main problem. After providing the big picture, we then supply the subroutines with all details and correctness proofs separately.

We demonstrate the behavior of our algorithm experimentally in Section~\ref{sec:empirical}.
Moreover, in Section~\ref{sec:uniformAttacker}, we consider the random attack adversary which is less predictable than the maximum carnage adversary.
There we show how to adapt our algorithm with a few minor tweaks to cope with this more complex scenario.
We conclude in Section~\ref{sec:conclusion} where we discuss possible future research directions.

%% file: Model.tex

\section{Model}
\label{sec:model}
We work with the strategic network formation model proposed by Goyal et al.~\cite{GJKKM15,GJKKM15arxiv} and mostly use their notation.
In this model the $n$ nodes of a network $G = (V, E)$ correspond to individual players $v_1,\dots,v_n$.
We will thus use the terms node, vertex and player interchangeably.
The edge set $E$ is determined by the players' strategic behavior as follows.
Each player $v_i \in V$ can decide to buy undirected edges to a subset of other players, paying $\alpha > 0$ per edge, where $\alpha$ is some fixed parameter of the model. 

If player $v_i$ decides to buy the edge to node $v_j$, then we say that the edge $\{v_i,v_j\}$ is owned and paid for by player $v_i$.\footnote{If both players $v_i$ and $v_j$ decide to buy the edge $\{v_i,v_j\}$, then this results in a multi-edge between $v_i$ and $v_j$. 
However, it is easy to see that best response strategies will never contain multi-edges which is why we ignore them completely.}
Buying an undirected edge entails connectivity benefits and risks for both participating endpoints. 
In order to cope with these risks, each player can also decide to buy immunization against attacks at a cost of $\beta > 0$, which is also a fixed parameter of the model.
We call a player \emph{immunized} if this player decides to buy immunization, and \emph{vulnerable} otherwise.

The strategy $s_i = (x_i,y_i)$ of player $v_i$ consists of the set $x_i \subseteq V \setminus \{v_i\}$ of the nodes to buy an edge to, and the immunization choice $y_i \in \{0, 1\}$, where $y_i = 1$ if and only if player $v_i$ decides to immunize.
The strategy profile $\mathbf{s} = (s_1, \ldots, s_n)$ of all players then induces an undirected graph 
$$G(\mathbf{s}) = \left(V,\bigcup_{v_i\in V}\bigcup_{v_j\in x_i} \{v_i,v_j\}\right).$$
The immunization choices $y_1,\dots,y_n$ in $\mathbf{s}$ partition $V$ into the set of immunized players $\mathcal{I} \subseteq V$ and vulnerable players $\mathcal{U} = V \setminus \mathcal{I}$.
The components in the induced subgraph $G[\mathcal{U}]$ are called \emph{vulnerable regions} and the set of those regions will be denoted by $\mathcal{R}_\mathcal{U}$.
The vulnerable region of any vulnerable player $v_i \in \mathcal{U}$ is $\mathcal{R}_\mathcal{U}(v_i)$.
\emph{Immunized regions} $\mathcal{R}_\mathcal{I}$ are defined analogously as the components of the induced subgraph $G[\mathcal{I}]$.

After the network $G(\mathbf{s})$ is built, we assume that an adversary attacks one vulnerable player according to a strategy known to the players. We consider mostly the maximum carnage adversary~\cite{GJKKM15arxiv,GJKKM15} which tries to destroy as many nodes of the network as possible.
To achieve this, the adversary chooses a vulnerable region of maximum size and attacks some player in that region.
If there is more than one such region with maximum size, then one of them is chosen uniformly at random.
If a vulnerable player $v_i$ is attacked, then $v_i$ will be destroyed and the attack spreads to all vulnerable neighbors of $v_i$, eventually destroying all players in $v_i$'s vulnerable region $\mathcal{R}_{\mathcal{U}}(v_i)$.
Let $t_{max} = \max_{R \in \mathcal{R}_\mathcal{U}} \{ |R| \}$ be the number of nodes in the vulnerable region of maximum size and $\mathcal{T} = \{ v_i \in \mathcal{U} \mid |\mathcal{R}_\mathcal{U}(v_i)| = t_{max} \}$ is the corresponding set of nodes which may be targeted by the adversary.
The set of targeted regions is $\mathcal{R}_\mathcal{T} = \{ R \in \mathcal{R}_\mathcal{U} \mid |R| = t_{max}\}$, and $\mathcal{R}_\mathcal{T}(v_i)$ is the targeted region of a player $v_i \in \mathcal{T}$.
Thus, if $v_i \in \mathcal{T}$ is attacked, then all players in the region $\mathcal{R}_{\mathcal{T}}(v_i)$ will be destroyed.

The \emph{utility} of a player $v_i$ in network $G(\mathbf{s})$ is defined as the expected number of nodes reachable by $v_i$ after the adversarial attack on network $G(\mathbf{s})$ (zero in case $v_i$ was destroyed) less $v_i$'s expenditures for buying edges and immunization.
More formally, let $CC_i(t)$ be the connected component of $v_i$ after an attack to node $v_t \in \mathcal{T}$ and let $|CC_i(t)|$ denote its number of nodes.
Then the utility (or profit)  $u_i(\mathbf{s})$ of $v_i$ in the strategy profile $\mathbf{s}$ is\vspace*{-0.2cm}
$$u_i(\mathbf{s}) = \frac{1}{|\mathcal{T}|} \left( \sum_{v_t \in \mathcal{T}} |CC_i(t)| \right) - |x_i| \cdot \alpha - y_i \cdot \beta.\vspace*{-0.2cm}$$
Fixing the strategies of all other players, the \emph{best response} of a player $v_i$ is a strategy $s_i^* = (x_i^*,y_i^*)$ which maximizes $v_i$'s utility $u_i\big((s_1, \ldots, s_{i-1}, s_i^*, s_{i+1}, \ldots, s_n)\big)$.
We will call the strategy change to $s_i^*$ a best response for player $v_i$ in the network $G(\mathbf{s})$, if changing from strategy $s_i \in \mathbf{s}$ to strategy $s_i^*$ is the best possible strategy for player $v_i$ if no other player changes her strategy.

Consider what happens if we remove node $v_i$ from the network $G(\mathbf{s}) = (V,E)$ and we call the obtained network $G(\mathbf{s}) \setminus v_i$.
In this case, $G(\mathbf{s}) \setminus v_i$ consists of connected components $C_1,\dots,C_\ell$.
The edge-set $x_i^*$ can thus be partitioned into $\ell$ subsets $x_i^*(C_1),\dots,x_i^*(C_\ell)$, where $x_i^*(C_z)$ denotes the set of nodes in $C_z$ to which $v_i$ buys an edge under best response strategy $s_i^*$.
We will say that $x_i^*(C_z)$ is an \emph{optimal partner set} for component $C_z$.
Therefore,  $x_i^*$ is the union of optimal partner sets for all connected components in $G(\mathbf{s}) \setminus v_i$.

A best response is calculated for one arbitrary but fixed player $\ap$, which we call the \textit{active player}. 
Furthermore let $\mathcal{C}$ be the set of connected components which exist in $G(\mathbf{s}) \setminus \ap$. Let 
\begin{align*}
 \CU &= \{ C \in \mathcal{C} \mid C \cap \mathcal{I} = \emptyset \},\\
 \CI &= \mathcal{C} \setminus \CU, \\
 \Cinc &= \{ C \in \mathcal{C} \mid \exists u \in C : \{u,v\} \in E \},
\end{align*}
where $\CU$ is the set of components in which all vertices are vulnerable,
$\CI$ is the set of components which contain at least one immunized vertex and $\Cinc$ is the set of components to which player $\ap$ is connected through incoming edges bought by some other player.

%% file: Algorithm.tex

\section{The Best Response Algorithm}\label{sec:algorithm}

\input{Algorithm-introduction}
\input{Algorithm-main}

\input{Algorithm-metatree-select}

%% file: Algorithm-introduction.tex
A naive approach to calculate the best response for player $\ap$ would consider all $2^n$ possible strategies and select one that yields the best utility.
This is clearly infeasible for a larger number of players.
\subsection{Key Observations}
Our algorithm exploits three observations to reduce the complexity from exponential to polynomial:

\textbf{Observation 1:} The network $G(\mathbf{s}) \setminus \ap$ may consist of $\ell$ connected components that can be dealt with independently for most decisions.
As long as the set of possible targets of the adversary does not change, the best response of $\ap$ can be constructed by first choosing components to which a connection is profitable and then choosing for each of those components an optimal set of nodes within the respective component to build edges to.
	
\textbf{Observation 2:} Homogeneous components in $G(\mathbf{s}) \setminus \ap$, which consist of only vulnerable or only immunized nodes, provide the same benefit no matter whether $\ap$ connects to them with one or with more than one edge. Thus the connection decision is a binary decision for those components.

\textbf{Observation 3:} Mixed components in $G(\mathbf{s}) \setminus \ap$, which contain both immunized and vulnerable nodes, consist of homogeneous regions that again have the property that at most one edge per homogeneous region can be profitable.
Merging those regions into block nodes forms an auxilliary tree, called \metatree, which we use in an efficient dynamic programming algorithm to compute the most profitable subset of regions to connect with.

Using the above observations, we construct an algorithm with a worst-case run time complexity of $\mathcal{O}(n^4 + k^5)$ where $n$ is the number of nodes of the network and $k$ the number of block nodes in the largest \metatree.

%% file: Algorithm-main.tex
\subsection{Main Algorithm}
\label{subsec:mainalgo}

\newcommand{\CSet}{\mathcal{A}}

In this section, we introduce our main algorithm, called \textsc{BestResponseComputation}. Its pseudo code can be found in Algorithm~\ref{alg:bestResponseMain} and a schematic overview is depicted in Fig.~\ref{fig:overview}. 
\begin{figure}[h!]
	\centering
	\includegraphics[width=9cm]{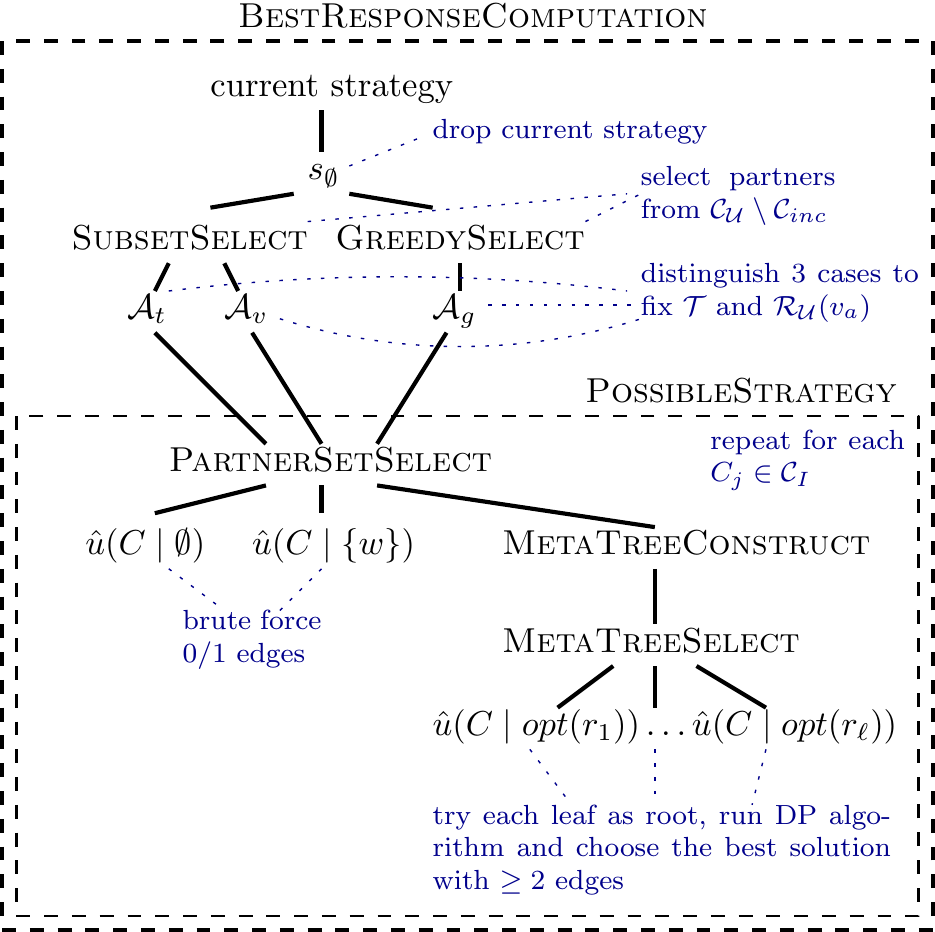}
	\caption{Schematic overview of the best response algorithm.}
	\label{fig:overview}
	
\end{figure}

Our algorithm solves the problem of finding a best response strategy by considering both options of buying or not buying immunization and computing for both cases the best possible set of edges to buy.
Thus, the first step of \textsc{BestResponseComputation} is to drop the current strategy of the active player $\ap$ and to replace it with the empty strategy $s_\emptyset = (\emptyset,0)$ in which player $\ap$ does not buy any edge and does not buy immunization.
Then the resulting strategy profile $\mathbf{s'} = (s_1, \ldots, s_{\aps-1}, s_\emptyset, s_{\aps+1}, \ldots, s_n)$ and the set of connected components $\CU$ and $\CI$ with respect to network $G(\mathbf{s'})\setminus \ap$ is considered.

The subroutine \textsc{SubsetSelect}, which is described in Section~\ref{sec:2D-KP}, determines the optimal sets of components of $\CU$ to connect to if $\ap$ does not immunize. This is done by solving an adjusted Knapsack problem which involes only small numbers. Two such sets of components, called $\mathcal{A}_t$ and $\mathcal{A}_v$, are computed depending on player $\ap$ becoming targeted or not by connecting to these components.
Additionally the subroutine \textsc{GreedySelect}, which is described in Section~\ref{sec:Greedy-Select}, computes a best possible subset of components of $\CU$ to connect with in case $\ap$ buys immunization by greedily connecting to all components in $\CU$ that are profitable.

The challenging part of the problem is to cope with the connected components in $\CI$ which also contain immunized nodes.
For such components our algorithm detects and merges equivalent nodes and thereby simplifies these components to an auxiliary tree structure, which we call the \metatree.
This tree is then used in a dynamic programming fashion to efficiently compute the best possible set of edges to buy towards nodes within the respective component.
Thus, our approach for handling components containing immunized nodes can be understood as first performing a data-reduction similar to many approaches for kernelization in the realm of Parameterized Algorithmics~\cite{DF13} and then solving the reduced problem via dynamic programming. 
\begin{center}
\begin{minipage}{14cm}
\begin{algorithm}[H]
    \KwIn{Strategies $\mathbf{s} = (s_1, \ldots, s_n)$, Player $\ap, 1 \leq \aps \leq n$}
    \KwOut{Best response strategy of player $\ap$ denoted by $s_\aps = (x_\aps, y_\aps)$}

    $s_\emptyset = (\emptyset, 0)$\;
    Let $G(\mathbf{s'})$ be the induced game state with $\mathbf{s'} = (s_1, \ldots, s_{\aps-1}, s_\emptyset, s_{\aps+1}, \ldots, s_n)$\;

    Let $\CSet_t, \CSet_v$ be the solutions of \textsc{SubsetSelect} on $\CU$\;
    Let $\CSet_g$ be the solution of \textsc{GreedySelect} on $\CU$\;

    $s_t$ = \textsc{PossibleStrategy}($\CSet_t$, 0)\;
    $s_v$ = \textsc{PossibleStrategy}($\CSet_v$, 0)\;
    $s_g$ = \textsc{PossibleStrategy}($\CSet_g$, 1)\;

    $S = \{s_{\emptyset}, s_t, s_v, s_g\}$\;
    \Return{\textnormal{strategy $s \in S$ which maximizes $\ap$'s utility}}\;

    \caption{\textsc{BestResponseComputation}}
    \label{alg:bestResponseMain}
\end{algorithm}
\end{minipage}
\end{center}

The subroutine \textsc{PossibleStrategy}, see Algorithm~\ref{alg:possiblepartners}, obtains the best set of nodes in components in $\CI$.
As this set depends on the number of targeted regions, it has to be determined for several cases independently.
These cases are $\ap$ not being immunized and not being targeted, $\ap$ not being immunized but being targeted, and $\ap$ being immunized.
For each case, the subroutine \textsc{PossibleStrategy} first chooses an arbitrary single edge to buy into the previously selected components from $\CU$.
\begin{center}
\begin{minipage}{15.5cm}
\begin{algorithm}[H]
        \KwOut{Best strategy with single edges to components in $\mathcal{A}$, given immunization $y_\aps$}
        $M := \emptyset$\;
        \ForEach{$C \in \mathcal{A}$} {
            $M = M \cup v_i$ for an arbitrary node $v_i \in C$\;
        }

        Locally, add edges to nodes in $M$, update $\mathcal{R}_\mathcal{I}, \mathcal{R}_\mathcal{U}, \mathcal{R}_\mathcal{T}$ according to $y_\aps$ and $M$\;
        $B \leftarrow \emptyset$\;
        \ForEach{component $C \in \CI$}{
            $B \leftarrow B ~\cup$ \textsc{PartnerSetSelect}$(C)$
        }
        \Return{($M \cup B , y_{\aps}$)}\;

    \caption{\textsc{PossibleStrategy}($\mathcal{A}$, $y_{\aps}$)}
    \label{alg:possiblepartners}
\end{algorithm}    
\end{minipage}
\end{center}
Then the best set of edges to buy into components in $\CI$ is computed independently for each component $C\in \CI$ via the subroutines \textsc{PartnerSetSelect} (see Section~\ref{algo:partnersetselect}), \textsc{MetaTreeConstruct} (see Section~\ref{sec:mt-construction}) and \textsc{MetaTreeSelect} (see Section~\ref{subsubsec:twoNewConnections}).
The union of the obtained sets is then returned.
Finally, the algorithm compares the empty strategy and the individually obtained best possible strategies for the above mentioned cases and selects the one which maximizes player $\ap$'s utility.

We establish that all the mentioned subroutines work together as desired in Section~\ref{sec:combination}.
We specify the constructed \metatree in Section~\ref{sec:mt-construction} and the \textsc{MetaTreeSelect} algorithm in Section~\ref{subsubsec:twoNewConnections}.

The run time of our best response algorithm heavily depends on the size of the largest obtained \metatree and we achieve a worst-case run time of
$\mathcal{O}(n^4 + k^5)$ for the maximum carnage adversary and $\mathcal{O}(n^4 + nk^5)$ for the random attack adversary, where $n$ is the number of nodes in the network and $k$ is the number of blocks in the largest \metatree (See Section~\ref{sec:costanalysis} and Section~\ref{sec:uniformAttacker}). In the worst case, this yields a run time of $\mathcal{O}(n^5)$ and $\mathcal{O}(n^6)$, respectively.
To contrast this worst-case bound, we also provide empirical results in Section~\ref{sec:empirical} showing that $k$ is usually much smaller than $n$, which emphasizes the effectiveness of our data-reduction and thereby shows that our algorithm is expected to be much faster than the worst-case upper bound.

\subsection{Correctness of \textsc{PossibleStrategy} and  \textsc{Best\-ResponseComputation}}
\label{sec:combination}

We now prove the correctness of \textsc{PossibleStrategy} and \textsc{BestResponseComputation} under the assumption that the subroutines \textsc{SubsetSelect}, \textsc{Greedy\-Select} and \textsc{PartnerSetSelect} are correct.

It is important to note, that \textsc{PartnerSetSelect} will not only return a best set of edges to buy into a component $C\in \CI$ but that all these edges connect to immunized nodes (see Lemma~\ref{lm:onlyImmunized}).
Thus, these edges do not change the set of targeted nodes $\mathcal{T}$ or player $\ap$'s vulnerable region $\mathcal{R}_{\mathcal{U}}(\ap)$.   

We start with proving the correctness of subroutine \textsc{PossibleStrategy}.
This subroutine gets a set of components from $\CU$ as input and then chooses an arbitrary node from each component to buy an edge to.
First, we show that this is correct.
\begin{restatable}{lemma}{restateMaxOneEdge}
	\label{lm:maxOneEdge}
	Buying at most one edge into any component $C \in \CU$ yields maximum profit for player $\ap$.
\end{restatable}
\begin{proof}
	Observe that in the case of an attack on component $C$ it is destroyed completely, if the attack takes place elsewhere $C$ stays connected.
	Hence, a single edge to any player in $C$ provides connectivity to all surviving players in $C$ in all cases, more edges therefore would only increase expenses but not connectivity.
\end{proof}

\noindent Thus, the optimal choice of edges into components in $\CU$ only depends on the right choice of the components by \textsc{SubsetSelect} and \textsc{GreedySelect}.

For components in $\CI$ we assume that \textsc{PartnerSetSelect} is correct and that buying edges into components in $\CI$ do not change $\mathcal{T}$ or $\mathcal{R}_{\mathcal{U}}(\ap)$.
We therefore only have to prove that under this assumption it is indeed correct to handle all components in $\CI$ independently.
For establishing this result, we first have to define player $\ap$'s expected profit contribution of a single component $C\in \CI$.

\subsubsection{Calculation of the Expected Profit Contribution of a Single Component}
\label{subsubsec:partialProfitContribution}

Assume that the adversary attacks vertex $t$.
In this case, by abusing notation, let $CC_\ap(t) \cap C$ denote the set of nodes in component $C$ that are still connected to player $\ap$ after the attack on vertex $t$ and therefore contribute to $\ap$'s profit.
In order to compute the expected profit contribution of a component $C \in \CI$ for player $\ap$, we simply have to take the expectation of $|CC_\ap(t) \cap C|$ over all possible choices for an attack.

Let $\hat{u}_\ap(C \mid \Delta)$ be the profit contribution for player $\ap$ of a component $C \in \CI$ if $\ap$ buys edges to all nodes in the set $\Delta$.
Thus,
$$
\hat{u}_\ap(C \mid \Delta) = \frac{1}{|\mathcal{T}|}\sum_{t \in \mathcal{T}} \left|CC_\ap(t) \cap C \right| - \alpha |\Delta|.
$$

\subsubsection{Conditional Independence within $\CI$}\label{subsubsec_proof_component_independence}

We proceed to proving that handling components in $\CI$ independently is indeed correct.
\begin{restatable}{lemma}{restateDealIndependently}
	\label{lm:dealIndependently}
	Player $\ap$ can deal with distinct components from $\CI$ independently, if $\mathcal{T}$ and $\mathcal{R}_\mathcal{U}(\ap)$ do not change.
\end{restatable}
\begin{proof}
	Recall that
	$$
	u_\ap({\bf s}) = \frac{1}{\mathcal{T}} \sum_{t \in \mathcal{T}} |CC_\ap(t)| - |x_\aps|\alpha - y_\aps \beta.
	$$
	Every summand can be split according to the components that player $v$ can buy edges to.
	Let $\Delta_C$ be the edges $\ap$ buys into $C$. Then
	\begin{align*}
	u_\ap({\bf s}) = \frac{1}{\mathcal{T}}  &\sum_{t \in \mathcal{T}} \Big(\sum_{C \in \CI } \big( |CC_\ap(t) \cap C| - |\Delta_C|\alpha\big)\\
	 + &\sum_{C \in \CU } \big( |CC_\ap(t) \cap C| - |\Delta_C|\alpha\big)\Big) - y_\aps \beta.
	\end{align*}
	This can be written as
	\begin{align*}
	u_\ap({\bf s}) =& \sum_{C \in \CI } \left( \frac{1}{\mathcal{T}} \sum_{t \in \mathcal{T}} \big|CC_\ap(t) \cap C\big| - |\Delta_C|\alpha \right)\\ &+
	\frac{1}{\mathcal{T}} \sum_{t \in \mathcal{T}} \left( \sum_{C \in \CU } \left( \big|CC_\ap(t) \cap C\big| - |\Delta_C|\alpha \right) \right)
	- y_\aps \beta \\
	=& \sum_{C \in \CI } \hat{u}_\ap(C \mid \Delta_C)\\
	&+
	\underbrace{\frac{1}{\mathcal{T}} \sum_{t \in \mathcal{T}} \left( \sum_{C \in \CU } \left( \big|CC_\ap(t) \cap C\big| - |\Delta_C|\alpha \right) \right)
		- y_\aps \beta}_{(**)},
	\end{align*}
	using linearity of expectation and the definition of $\hat{u}$.
	Since both $\mathcal{T}$ and $\mathcal{R_U}(\ap)$ remain unchanged,
	$(**)$ is fixed.
	
	The choice of nodes to connect to within any component $C$ does not
	influence $\ap$'s connectivity to nodes outside of $C$ as every path
	from some node in $C$ to some node outside $C$ traverses node
	$\ap$. Thus, the choice of nodes in $C$ is independent from the choice
	of nodes in other components. Hence, the first summation's terms are independent.
	
	Because every term is non-negative, the summands may be maximized
	independently of each other to find the best response.
\end{proof}

\noindent Lemma~\ref{lm:maxOneEdge} and Lemma~\ref{lm:dealIndependently} directly yield the following:
\begin{corollary}
	If \textsc{PartnerSetSelect} is correct and selects only edges to immunized nodes, then \textsc{Possible\-Stra\-tegy} is correct.
\end{corollary}

\subsubsection{Correctness of \textsc{BestResponseComputation}}
Now we show that the main subroutines work together as desired. 
\begin{restatable}{theorem}{restateAlgmaincorrect}\label{algmaincorrect}
	If the algorithms \textsc{SubsetSelect}, \textsc{Greedy\-Select} and \textsc{PossibleStrategy} are correct, then \textsc{Best\-ResponseComputation} is correct.
\end{restatable}
\begin{proof}
	We assume that \textsc{PossibleStrategy} computes an optimal set of edges into components in $\CI$ and that all those edges connect to immunized nodes.  
	
	Firstly, observe that if $\ap$ decides to immunize, then the player cannot change $\mathcal{T}$ or $\mathcal{R}_{\mathcal{U}}(\ap)$ by buying edges.
	Moreover, in this case \textsc{GreedySelect} computes the optimal set of components in $\CU$ to connect with via a single edge.
	Thus, strategy $s_g$ (see Alg.~\ref{alg:bestResponseMain}) is indeed a best possible strategy for this case.
	
	The only changes to $\mathcal{T}$ or $\mathcal{R}_{\mathcal{U}}(\ap)$ can happen in case $\ap$ is vulnerable and buys edges to one or more vulnerable nodes.
	After $\ap$'s edge purchases we distinguish the following cases:
	
	\textbf{Case 1:} The new vulnerable region of $\ap$ is strictly smaller than $t_{max}$:
	The set of targeted nodes $\mathcal{T}$ does not change and $\ap$ remains un-targeted.
	By assumption \textsc{SubsetSelect} has chosen the optimal set of components $\mathcal{A}_v$ from $\CU$ and thus 
	strategy $s_v$ (see Alg.~\ref{alg:bestResponseMain}) is the best strategy in this case.
	
	\textbf{Case 2:} The new vulnerable region of $\ap$ has the size of exactly $t_{max}$:
	This increases the number of targeted nodes $|\mathcal{T}|$ by $t_{max}$ as $\ap$ becomes targeted.
	By assumption \textsc{SubsetSelect} has chosen the optimal set of components $\mathcal{A}_t$ from $\CU$ and thus
	strategy $s_t$ (see Alg.~\ref{alg:bestResponseMain}) is the best strategy in this case.
	
	\textbf{Case 3:} The new vulnerable region of $\ap$ is larger than $t_{max}$:
	In this case $\ap$ will be part of the single largest vulnerable region and thus killed.
	Clearly, in this case the only and best strategy for $\ap$ is to not buy any edges.
	
	Thus, \textsc{BestResponseComputation} covers all relevant cases and is therefore correct.
\end{proof}

\subsection{Subroutines Dealing with Components in $\CU$}
Components in $\CU$ are handled by the algorithms \textsc{SubsetSelect} and \textsc{GreedySelect}. The details of these algorithms are specified in this section.

\input{Algorithm-2DKP}
\input{Algorithm-greedy}

%% file: Algorithm-2DKP.tex
\subsubsection{SubsetSelect: Interdependent Subset Selection of $\CU$}
\label{sec:2D-KP}

We assume, the best response of $\ap$ does not include immunization.
By Lemma~\ref{lm:maxOneEdge} we know that the problem of finding optimal edge endpoints in components in $\CU$ is equivalent to the problem of finding a subset of components in $\CU$ that is most profitable.
In particular, it cannot be beneficial for $\ap$ to buy edges to components from $\CU \cap \Cinc$, as she is already connected to those.
Thus, $\CU \setminus \Cinc = \{C_1, \ldots, C_m\}$ are the vulnerable components $\ap$ might buy an edge to.

The algorithm \textsc{SubsetSelect} computes the subset of the set of components $\{C_1, \ldots, C_m\}$ which later yields the highest utility, by solving an adjusted Knapsack problem with $|C_i|$ as profit and the total number of connected nodes and the total number of newly built edges as weight limits.

In the following, let $M$ be a $3$-dimensional table, where a cell $M[x,y,z]$ of the table stores the maximum number smaller or equal to $z$ of nodes that player $\ap$ can connect to, using only components from $\{C_1, \ldots, C_x\}$ and building at most $y$ edges in total.
As we only buy a maximum of one edge per component, the dimensions of $M$ are $n \times m \times n$.
The dynamic program therefore uses $\mathcal{O}(n^2m)$ space.

$M$ is computed as follows:
\begin{align*}
    M[0, \cdot, \cdot] &= M[\cdot, 0, \cdot] = M[\cdot, \cdot, 0] = 0,\\
    M[x,y,z] &= \begin{cases}
    M[x-1,y,z], &|C_x| > z,\\
    \max\bigg(|C_x| + M\big[x-1,y-1,z-|C_x|\big],\ M\big[x-1,y,z\big]\bigg), &|C_x| \leq z.
    \end{cases}
\end{align*}
Let $r = t_{max} - |\mathcal{R}_\mathcal{U}(v)|$ be the remaining number of vulnerable nodes $\ap$ may connect to without forming a unique targeted region.
Then, the algorithm picks two solutions $a_t, a_v$ from $M$ defined as:
\begin{align*}
a_t &= \max_{0 \leq j \leq m}\{ M[m,j,r] - j \cdot \alpha \}\\
a_v &= \max_{0 \leq j \leq m}\{ M[m,j,r-1] - j \cdot \alpha \}
\end{align*}
Let $\mathcal{A}_t, \mathcal{A}_v \subseteq \{C_1, \ldots, C_m \}$ be the corresponding subsets to the solutions $a_t, a_v$, which can be constructed by tracking the predecessor of a field in $M$.

Hence, $a_t$ is the maximum benefit we gain by connecting to $|\mathcal{A}_t| \leq r$ additional nodes, $a_v$ is the maximum benefit with at most $|\mathcal{A}_v| < r$ additional nodes.
Thus, by using $\mathcal{A}_t$ for the solution, $\ap$ might become part of a vulnerable region of maximum size and the best response algorithm has to decide later, whether it is beneficial to risk being attacked.
By taking $\mathcal{A}_v$, the vulnerable region of $\ap$ will remain strictly smaller than $t_{max}$ and thus $\ap$ won't be targeted by the adversary.

The algorithm has to store both solutions to decide later, which one is better, as the additional risk has influence on the total utility of the player.

If $a_t \leq 0$ and $a_v \leq 0$, then the best response of node $\ap$ does not include edges to components in $\CU$ under the condition that $\ap$ does not immunize.

We omit the straightforward correctness proof.

%% file: Algorithm-greedy.tex
\subsubsection{GreedySelect: Greedy Subset Selection of $\CU$}
\label{sec:Greedy-Select}

If $\ap$ immunizes, she does not incur any risk by buying new edges into components from $\CU$.
Thus, by Lemma~\ref{lm:maxOneEdge} she can buy single edges to all components $C \in \CU \setminus \Cinc$ which contribute more to her expected profit than the edge cost $\alpha$.

Therefore the \textsc{GreedySelect} algorithm computes the set $\mathcal{A}_g$ which is the set of vulnerable components player $\ap$ buys an edge to if she immunizes:
\begin{align*}
    \mathcal{A}_g &= \left\{ C \in  \CU \setminus \Cinc \mid |C| \cdot p_{survive}(C) > \alpha \right\}\\
    p_{survive}(c) &= 1 - \frac{|C \cap \mathcal{T}|}{|\mathcal{T}|}.
\end{align*}
We omit the straightforward correctness proof.

%% file: Algorithm-metatree-select.tex
\subsection{Partner Selection for Components in $\CI$}
\label{subsec:mt-Select}

\newcommand{\uv}{v_\mathcal{U}}
\newcommand{\iv}{v_\mathcal{I}}

\newcommand{\URs}{\mathcal{R}_\mathcal{U}^C}
\newcommand{\IRs}{\mathcal{R}_\mathcal{I}^C}

\newcommand{\UR}{R_\mathcal{U}}
\newcommand{\IR}{R_\mathcal{I}}

Let $C_1, \ldots, C_c \in \CI$ be the components $\ap$ might buy edges into.
By definition, each of those components contains at least one immunized node.
Lemma~\ref{lm:onlyImmunized} (see Section~\ref{sec:mt-properties}) ensures that we only need to consider buying edges to immunized nodes.

For computing an optimal partner set for a component $C\in \CI$, we consider the expected contribution of $C$ to $\ap$'s profit given that $\ap$ buys edges to all nodes in a set $\Delta$, and denote this profit by $\hat{u}_\ap(C \mid \Delta)$.
The details of $\hat{u}$ are defined in Section~\ref{subsubsec:partialProfitContribution}.

\subsubsection{PartnerSetSelect} \label{algo:partnersetselect}
For each component $C \in \CI$ we compute three candidate sets of players to buy edges to and finally select the candidate set that yields the highest profit contribution for the considered component $C$ for player $\ap$.
The three candidate sets for component $C$ are obtained as follows:

\textbf{Case 1:} The player considers buying no additional edges into component $C$.
	In this case the resulting player set is empty.
	
\textbf{Case 2:} The player considers buying exactly one additional edge into component $C$.
	The resulting player set contains the immunized partner that maximizes the profit for this component.
	
\textbf{Case 3:} The player considers buying at least two edges.
	In this case an optimal set of at least two immunized partners is obtained using the algorithm \textsc{MetaTreeSelect}.

As all possible cases are covered, the most profitable set of those three candidate solutions for component $C$ must be the optimal partner set for component $C$.
This optimal partner set is returned.
We refer to this subroutine as \textsc{PartnerSet\-Select}.

The first two cases, buying either no or exactly one edge into component $C$ are easily solved:
if no edge is purchased by $\ap$, then the expected profit contribution is $\hat{u}_\ap(C \mid \emptyset)$.
If exactly one edge is bought then the expected profit contribution is $\hat{u}_\ap(C \mid \{w\})$, where $w$ is the vertex in $C$ which maximizes $\ap$'s expected profit for component $C$.

Case 3 is much more difficult to handle.
It is the main point where we need to employ algorithmic techniques to avoid a combinatorial explosion.
To ease the strategy selection, for each component $C \in \CI$ we create an auxiliary graph to identify sets of nodes which offer equivalent benefits with respect to connection.
This graph is a bipartite tree which we call the \metatree of $C$.
Fig.~\ref{fig:sample-metas} shows a conversion of a graph component into its \metatree by merging adjacent nodes of the same type into regions and collapsing regions into blocks.
So called \bridgeblocks (orange) of the \metatree represent targeted regions of $C$ that would, if destroyed, decompose $C$ into at least two components.
If the adversary however chooses to attack a player in a so-called \candidateblock (blue or violet), $C$ would remain connected.
The construction procedure can be found in Section~\ref{sec:mt-construction} and useful properties of the \metatree are proven in Section~\ref{sec:mt-properties}.

\begin{figure}[htb]
	\centering
	\includegraphics[width=4cm]{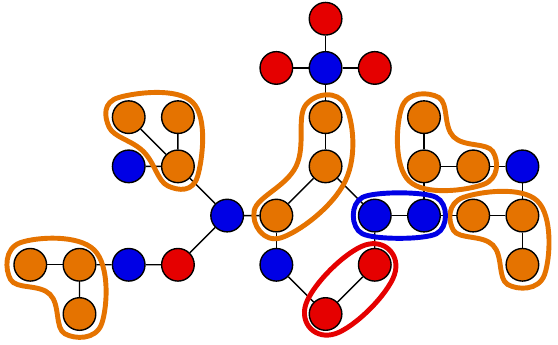}
	\hfill
	\includegraphics[width=4cm]{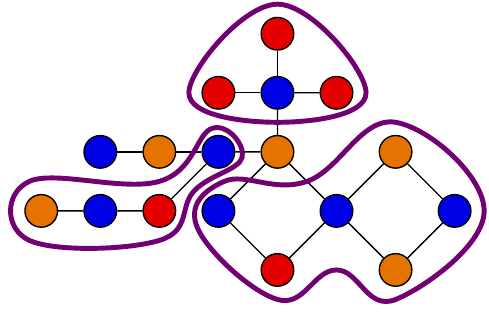}
	\hfill
	\includegraphics[width=2.7cm]{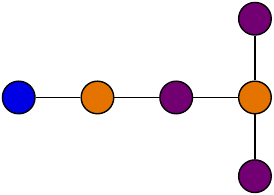}
	
	~
	
	\hfill
	\includegraphics{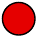} {\small vulnerable}
	\hfill
	\includegraphics{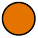} {\small targeted}
	\hfill
	\includegraphics{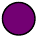} {\small mixed}
	\hfill
	\includegraphics{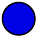} {\small immunized}
	\hfill
	\caption{A graph component (left), the corresponding \metagraph (middle), which is an intermediate step in the construction, and the obtained \metatree (right).}
	\label{fig:sample-metas}
\end{figure}

\noindent We now use the \metatree to identify the strategy maximizing the expected profit for $\ap$.

\subsubsection{\metatree Construction}
\label{sec:mt-construction}

We now describe the algorithm \textsc{MetaTreeConstruct} for constructing the \metatree of component $C$. 

To construct the \metatree, we first create the \metagraph of $C$ by merging maximally connected subcomponents of $C \cap \mathcal{U}$ and $C \cap \mathcal{I}$, respectively.
Thus, the resulting undirected \metagraph $G' = (V', E')$ is bipartite with $V' = \mathcal{R}_\mathcal{U}^C \cup \mathcal{R}_\mathcal{I}^C$, $\mathcal{R}_\mathcal{U}^C = \{R \in \mathcal{R}_\mathcal{U} \mid R \subseteq C \}$, $\mathcal{R}_\mathcal{I}^C = \{R \in \mathcal{R}_\mathcal{I} \mid R \subseteq C \}$ and $E' = \{ \{\UR,\IR\} \in \mathcal{R}_\mathcal{U}^C \times \mathcal{R}_\mathcal{I}^C \mid \exists v_\mathcal{U} \in \UR, v_\mathcal{I} \in \IR : \{v_\mathcal{U},v_\mathcal{I}\} \in E \}$.

To create the \metatree from the \metagraph, we identify vertices in $G'$ that disconnect the \metagraph (and thus, the original graph) if the adversary attacks any player inside the region represented by these vertices.
We name these vertices \bridgeblocks ($BB$).
Regions that stay connected no matter which player the adversary targets are collapsed into so called \candidateblocks ($CB$), to which the active player might buy edges.
Those blocks are constructed iteratively:

    \textbf{(1)} We create a new \candidateblock $CB_x$ by identifying any immunized region $\IR \in \IRs$ which is not yet part of a \candidateblock:
    $CB_x = \{\IR\}$ with $\IR \in \IRs$ and $\IR$ is not yet part of a block.
    
    \textbf{(2)} We add all immunized regions $R$ that are reachable via two paths originating from any $R' \in CB_x$, which do not share a targeted region:
    \begin{align*}CB_x &= CB_x \cup \{R \in \IRs \mid \exists R' \in CB_x ~\exists P = R' \rightarrow R,\\
    & Q = R' \rightarrow R : (P \cap Q) \cap \mathcal{R}_\mathcal{T} = \emptyset\},\end{align*}
    where $R' \to R$ denotes a simple path from vertex $R'$ to vertex $R$.
    Note that the paths can be identical, thus two immunized regions that are connected through a vulnerable but not targeted region are collapsed as well.
    
    \textbf{(3)} We add any nodes that only connect to nodes already in $CB_x$:
    $CB_x = CB_x \cup \{\UR \in \URs \mid \forall \{\UR, R\} \in E' : R \in CB_x\}.$
    This especially merges any non-targeted vulnerable regions into the respective adjacent \candidateblock.\footnote{Note that there is always exactly one adjacent \candidateblock.}

After we have constructed all \candidateblocks, i.e. when all immunized nodes are part of a \candidateblock, any remaining nodes become \bridgeblocks $BB_x$.

Again, if an edge existed between two nodes that are now in separate blocks, these blocks are connected in the \metatree $M = (V^*, E^*)$ with
$$V^* = \{ CB_1, ..., CB_i, BB_1, ..., BB_j \}$$
and
$$E^* = \{ \{BB_x, CB_y\} \mid \exists \{\UR, \IR\} \in E' : \UR \in BB_x \wedge \IR \in CB_y \}.$$
Since the underlying \metagraph is bipartite and no additional edges are introduced, the \metatree is also bipartite.

\begin{restatable}{lemma}{restateMtTree}
    The \metatree M is indeed a tree.
    \label{lm:mt-tree}
\end{restatable}
\begin{proof}
	Our algorithm starts on a graph which is connected and only collapses nodes.
	This way no node can get disconnected as edges are only removed when the nodes they connected are merged.
	M is therefore connected.
	
	We now use the fact that M is bipartite by construction to prove by contradiction that M contains no cycles.
	Assume there exists a cycle in the \metatree.
	As M is bipartite, the cycle must have an even length $>2$, as we allow no multi-edges.
	Thus the cycle contains at least two \candidateblocks $CB_1 \neq CB_2$.
	These two blocks contain two immunized regions $\IR \in CB_1, \IR' \in CB_2$.
	As they are arranged in a cycle, there exist two paths $P = \IR \rightarrow \IR', Q = \IR \rightarrow \IR'$ that do not share a targeted region, thus according to the construction $\IR$ and $\IR'$ must be part of the same \candidateblock.
	This contradicts $CB_1 \neq CB_2$.
\end{proof}

\begin{restatable}{lemma}{restateMtLeafIsCB}
    After the construction of the \metatree all leaves are \candidateblocks.
    \label{lm:mtLeafIsCB}
\end{restatable}
\begin{proof}
	Assume the contrary.
	Then there exists a \bridgeblock leaf $BB$.
	Because the \metatree is a bipartite tree and the underlying \metagraph contains at least one immunized region, the parent of this leaf has to be a \candidateblock $CB \neq BB$.
	Then $BB$ is only connected to exactly one \candidateblock and must by construction be merged into $CB$.
	Then however $BB$ would not be a leaf which is a contradiction, thus all leaves are \candidateblocks.
\end{proof}

\subsubsection{Key Properties of Using the \metatree}
\label{sec:mt-properties}
\begin{restatable}{lemma}{restateOnlyImmunized}
	\label{lm:onlyImmunized}
	Player $\ap$ has an optimal partner set for $C \in \CI$ which only buys edges to immunized players.
\end{restatable}
\begin{proof}
    Consider a best response for component $C \in \CI$ for player $\ap$ which contains buying an edge to some vulnerable node $\uv \in C$.
    Since $C \in \CI$, at least one immunized node $\iv$ must be adjacent to $\uv$'s vulnerable region $\mathcal{R}_\mathcal{U}(\uv)$.

    On the one hand, if $\mathcal{R}_\mathcal{U}(\uv)$ is attacked, then player $\ap$ gets no benefit from buying the edge to $\uv$.
    Buying an edge to $\iv$ instead, would yield at least the same benefit.
    On the other hand, if any other component $C_w \in \mathcal{R}_\mathcal{T} \setminus \mathcal{R}_\mathcal{U}(\ap)$ is attacked, then the vertices $\uv$ and $\iv$ stay connected and thus an edge to $\iv$ provides the same benefit as an edge to $\uv$.

    Hence, each edge to some vulnerable player $\uv$ can be exchanged with an edge to an immunized player $\iv$ that is adjacent to $\mathcal{R}_\mathcal{U}(\uv)$ without decreasing the expected profit.
\end{proof}

\noindent We show that in case of buying at least two edges into $C\in \CI$, player $\ap$ only has to consider leaves of the \metatree of $C$ as endpoints.
Remember that all vertices of the \metatree are either \candidateblocks or \bridgeblocks and that by Lemma~\ref{lm:mtLeafIsCB}, all leaves of the \metatree are \candidateblocks.
Since \candidateblocks may contain many nodes, we first have to clarify what it means to buy an edge to a \candidateblock.
We show that if a player wants to buy an edge into a \candidateblock, then buying a single edge to an immunized node of the \candidateblock suffices.

\begin{restatable}{lemma}{restateOneEdgeToCB}
	It is never beneficial to buy more than one edge into a \candidateblock.
	\label{lm:oneEdgeToCB}
\end{restatable}
\begin{proof}
By Lemma~\ref{lm:onlyImmunized}, it suffices to buy ed\-ges to immunized nodes of the \candidateblock.
For any pair of immunized nodes $\iv, \iv'$ in the \candidateblock, one of the following statements, directly resulting from properties of the \metatree, will always be true:
\begin{enumerate}
	\item There exists a direct edge or a path only containing immunized nodes between $\iv$ and $\iv'$.
	\item For any targeted region $R$ in the \candidateblock there exists a path connecting $\iv$ and $\iv'$ which does not contain $R$.
\end{enumerate}
Thus, no matter which node will be attacked, each immunized node in the \candidateblock will always be connected to the whole \candidateblock except for the nodes that get killed by the attack.
It follows, that buying at most one edge per \candidateblock suffices.
\end{proof}

\noindent Now we are ready to show that it is optimal to buy edges towards \candidateblocks which are leaves of the \metatree.
\begin{restatable}{lemma}{restateMtConnectToLeaves}
	\label{lm:mt-connect-to-leaves}
	Let $M$ be the \metatree of component $C$.
	If player $\ap$ has an optimal partner set for component $C$ which contains buying at least two edges, then $\ap$ also has an optimal partner set for component $C$ which contains only leaves of $M$.
\end{restatable}
\begin{proof}
By Lemma~\ref{lm:onlyImmunized} w.l.o.g. $\ap$ only buys edges to \candidateblocks in $M$.
Assume the best response strategy $s_\aps^*$ of $\ap$ connects to $M$ with at least two edges with at least one of the endpoints not being a leaf of the \metatree.
Let $CB$ be such an inner block, root the \metatree at $CB$ and inspect its subtrees.

If $\ap$ bought edges to at least two different subtrees, one of those edges will be connected to $CB$ no matter where the adversary attacks.
Thus, by dropping the edge to $CB$ the player does not reduce connectivity whilst decreasing her expenses, which is a contradiction to the assumption that the edge to $CB$ is part of a best response.

If however $\ap$ only bought edges to one of the subtrees of $CB$, $\ap$ can exchange the edge to $CB$ with an edge to a leaf $CB'$ in another subtree
Due to Lemma~\ref{lm:mtLeafIsCB} this leaf is a \candidateblock as well.
Again, independent of the attack target of the adversary, $\ap$ is connected to $CB$ through one of the subtrees she bought an edge to, hence the exchange does not decrease her utility.

Hence, each edge to some inner block $CB$ can be exchanged with an edge to a leaf block $CB'$ without decreasing the expected profit.
\end{proof}

\subsubsection{Solving Case 3 of \textsc{PartnerSetSelect}}
\label{subsubsec:twoNewConnections}
In the following let $M$ be the \metatree of component $C$.
Moreover, we assume that $M$ has at least two \candidateblocks, since otherwise, by Lemma~\ref{lm:onlyImmunized} (see Section~\ref{sec:mt-properties}), buying at most one edge suffices.

\paragraph{Idea of the \textsc{MetaTreeSelect} Algorithm}
By Lemma~\ref{lm:oneEdgeToCB} and Lemma~\ref{lm:mt-connect-to-leaves} (see Section~\ref{sec:mt-properties}) we only have to consider to buy single edges into leaves of the \metatree which are \candidateblocks.
Thus, we only have to find the optimal combination of leaves of $M$ to which to buy an edge.

Probing edge purchases to all possible combinations of leaves of $M$ and comparing the expected profit contributions yields a correct solution but entails exponentially many computations.
We reduce this exponential complexity by using the following observations to obtain the best possible combination of leaves.
Both observations are based on the assumption that player $\ap$ buys an edge to some leaf $r$ of $M$ and we consider the tree $M$ rooted at $r$.
Later we ensure this assumption by rooting $M$ at each possible leaf.

	\textbf{Observation 1:} Consider any vertex $w$ of $M$.
	If player $\ap$ has an edge to $w$, then it can be decided efficiently whether it is beneficial to buy exactly one or no edge into a subtree of $w$, as the influence of any additional edges into $M$ does not propagate over $w$, effectively making the decisions to buy edges independent for subtrees.
	
	\textbf{Observation 2:} Let $w$ be any vertex of $M$ and let the children of $w$ in $M$ be $x_1,\dots,x_\ell$.
	Consider that $\ap$ has an edge to $w$ and it has already been decided for each subtree rooted at $x_1,\dots,x_\ell$ whether or not to buy an edge into that subtree. If there exists at least one edge between $\ap$ and any of those subtrees, then it cannot be beneficial to buy additional edges into the subtree rooted at $w$. 
	Either $w$ is destroyed, and the previous edge-buy decisions apply to the disconnected subtrees, or $w$ survives, and $\ap$ is connected to all subtrees via node $w$.

These observations provide us with the foundation for a dynamic programming algorithm which decides bottom-up whether it is beneficial to buy at most one edge into a given subtree by reusing the edge buy decisions of its subtrees.

Note that the algorithm never has to compare combinations of bought edges, as the only decision to make is, whether or not to buy exactly one edge into a subtree in combination with iteratively shifting the presumed edge to the parent node of the leaves to the root $r$.

\paragraph{The \textsc{MetaTreeSelect} Algorithm} A detailed description of the \textsc{MetaTreeSelect} algorithm can be found in Algorithm~\ref{alg:metatree-select}.
\begin{center}
\begin{minipage}{15cm}
\begin{algorithm}[H]
    \KwIn{\metatree $M$ for component $C$}
    \KwOut{Player $\ap$'s optimal partner set for $C$ consisting of at least two partners}

    \ForEach{leaf $r$ of $M$} {
        $M(r) \leftarrow$ root $M$ at vertex $r$;
        
        $w \leftarrow r$'s only child in $M(r)$;

        $opt(r) \leftarrow \{\textnormal{some immunized node in }r\} ~\cup$ \textsc{RootedMetaTreeSelect}($M(r),w$);

    }

    $best \leftarrow opt(r)$ which maximizes $\hat{u}_\ap(C \mid opt(r))$;

    \If{$|best| \geq 2$} {
        \Return{$best$};
    } \Else {
        \Return{$\emptyset$};
    }

    \caption{\textsc{MetaTreeSelect}($M$)}
    \label{alg:metatree-select}
\end{algorithm}
\end{minipage}
\end{center}
The algorithm \textsc{MetatreeSelect} roots $M$ at every leaf and assumes buying an edge towards some immunized node within the root \candidateblock.
Then the subroutine \textsc{RootedMetaTreeSelect} (see Algorithm~\ref{alg:rooted-metatree-select}) gets the rooted \metatree $M(r)$ and some vertex $r_T$ (which initially is the only child of the currently considered root leaf) as input and recursively computes the expected profit contribution of one additional edge from $\ap$ to a block in the subtree $T$ rooted at $r_T$ under the assumption that $\ap$ is already connected to the parent block $p(r_T)$ of $r_T$ in $M(r)$.
\begin{center}
\begin{minipage}{14cm}
\begin{algorithm}[H]
    \KwIn{rooted \metatree $M(r)$, vertex $r_T$ of $M(r)$}
    \KwOut{Set of nodes from $T$ to buy an edge to}

    $opt(r_T) \leftarrow \emptyset$;

    \ForEach{child $w$ of $r_T$} {
        $opt(r_T) \leftarrow opt(r_T)~ \cup$ \textsc{RootedMetaTreeSelect}($M(r),w$);
    }

    \If{$r_T$ is a \bridgeblock \textbf{or} $opt(r_T) \neq \emptyset$ \textbf{or} a player in $T$ bought an edge to $\ap$} {
        \Return{$opt(r_T)$};
    }

    \ForEach{leaf $l$ of $T$} {
        $\profit(l) \leftarrow$ additional profit of $\ap$ with edge to $l$;
    }
    $best \leftarrow l$ which maximizes $\profit(l)$;

    \If{$\profit(best) > \alpha$} {
        $opt(r_T) \leftarrow opt(r_T)~ \cup \{\textnormal{some immunized node in }best\}$;
    }

    \Return{$opt(r_T)$};

    \caption{\textsc{RootedMetaTreeSelect}($M(r),r_T$)}
    \label{alg:rooted-metatree-select}
\end{algorithm}
\end{minipage}
\end{center}
Let $|T|$ denote the number of players represented by the union of all blocks in $T$ and $opt(r_T)$ will be the set of blocks in $T$ the algorithm decided to buy an edge to.

\noindent After processing all subtrees of $r_T$ (Alg.~\ref{alg:rooted-metatree-select} lines 2-3), the algorithm distinguishes three cases (Alg.~\ref{alg:rooted-metatree-select} line 4):

	\textbf{Case 1:} $r_T$ is a \bridgeblock.
	Then, as $M$ is bipartite, $p(r_T)$ must be a \candidateblock.
	As the algorithm assumes the existence of an edge from $\ap$ to $p(r_T)$, there also exists a path from $\ap$ to $r_T$ via $p(r_T)$ in all attack scenarios. Thus, no additional edge is needed (Alg.\ref{alg:rooted-metatree-select} line 5).
	
	\textbf{Case 2:} There exists an edge between $\ap$ and some node $x$ in $T$, either through an edge $\ap$ buys according to the results of the recursive invocations, or through a preexisting edge bought by player $x$.
	Then, depending on the attack target, there either exists a path from $\ap$ to $r_T$ via $x$ or via $p(r_T)$. Hence, no additional edge is needed (Alg.\ref{alg:rooted-metatree-select} line 5).
	
	\textbf{Case 3:} Player $\ap$ can get disconnected from $r_T$ by an attack on $p(r_T)$.
	Then the algorithm considers each leaf $l$ of $T$ as possible partner (Alg.\ref{alg:rooted-metatree-select} line 6), computes the profit contribution of an edge to $l$ (Alg.\ref{alg:rooted-metatree-select} line 7) and selects a leaf that maximizes this profit contribution (Alg.\ref{alg:rooted-metatree-select} line 8).

	The additional profit of an edge to $l$ is computed as follows: An edge to $l$ only yields profit, if a \bridgeblock $t$ is attacked which either belongs  to $T$ or $t=p(r_T)$, and $l$ is located in a subtree of $t$.
	In this case, the profit contribution equals the size of this subtree.
	Therefore let $\profit(l \mid t)$ be the additional profit an edge to $l$ contributes to the utility of $\ap$ in case $t$ is attacked and let $\mathcal{B}$ be the set of all \bridgeblocks in $T$.
	Thus
	$$\profit(l) = \frac{|p(r_T)|}{|\mathcal{T}|}|T| + \sum_{t \in \mathcal{B}}\frac{|t|}{|\mathcal{T}|} \profit(l \mid t),$$ with
	$$\profit(l \mid t) =
		\begin{cases}
			0 & l \text{ is not in any subtree of } t\\
			|Y| & Y \text{ is a subtree of } t \text{ and } l \text{ is in } Y.
		\end{cases}$$
	Finally, If the additional profit of the best possible leaf exceeds the edge costs, $l$ is added to the set of partners of $\ap$ (Alg.\ref{alg:rooted-metatree-select} line 10).

\subsubsection{Correctness of MetaTreeSelect and RootedMetaTreeSelect}\label{sec:mts-correctness}

\noindent An important detail for proving the correctness of \textsc{MetaTreeSelect} is the definition of $profit(l)$ from Case $3$ of \textsc{RootedMetaTreeSelect}.
There, the existence of an edge to $p(r_T)$ is assumed and the expected profit is computed by considering both the case where $p(r_T)$ is destroyed and the case where some other node is attacked.
Thus, while considering an assumed edge to some parent node, it is already taken into account that this node may be attacked.

Also note that if an edge to some node of the \metatree is bought then this means that an edge to some immunized node of the corresponding \candidateblock is bought.

The algorithm \textsc{MetaTreeSelect} computes an optimal partner set for component $C$ containing at least two nodes, if it exists.
In this case, by Lemma~\ref{lm:mt-connect-to-leaves}, there is an optimal partner set for $C$ containing only nodes within leaves of $M$.
Assume that the \metatree $M$ of component $C$ is rooted at such a leaf $r$.
Since the algorithm \textsc{MetaTreeSelect} compares all possibilities to root $M$ at a leaf we will assume in the following that some node $r^*$ in block $r$ is contained in an optimal partner set.

Fig.~\ref{fig:mt2Induc} shows the structure of the \metatree in the following proofs.
\begin{figure}[!h]
	\centering
	\includegraphics[width=6cm]{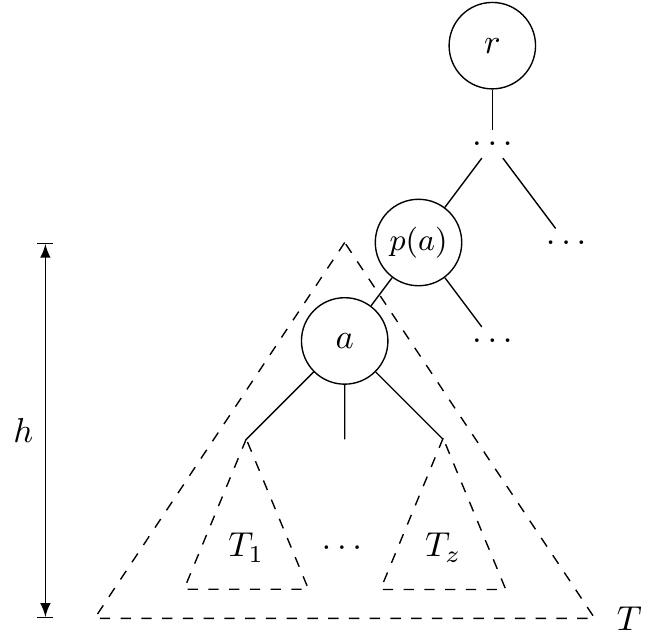}
	\caption{Inductive step of \textsc{RootedMetaTreeSelect}.}
	\label{fig:mt2Induc}
\end{figure}

\begin{lemma} \label{lm:noadditionaledges}
	Consider any subtree $T$ with height $h$ of $M$ and let $a$ be its root.
	If \textsc{RootedMeta\-Tree\-Select} has processed all subtrees $T_1, T_2, \ldots$ of $a$ under the assumption of an edge from $\ap$ to $a$ and at least one of those subtrees contains a node which was added to the partner set, then it cannot be profitable to buy additional edges into the complete tree $T$ under the assumption of the existence of an edge to $p(a)$.
\end{lemma}

\begin{proof}
	We prove the statement by induction over the height $h$.
	Assume there exists at least one edge between $\ap$ and some child $a'$ of $a$ and assume the existence of an edge to $p(a)$, the parent of $a$.
	Thus, there exists a path from $\ap$ to $a$ via $p(a)$.
	
	For $h=1$ all children of $a$ are leaves.
	The event of $a$ being destroyed was already considered in a previous recursive call.
	If some other node is attacked, then all surviving nodes in $T$ remain connected and $\ap$ has a path via some node in $a'$ to all of them.
	
	Assume now, that the hypothesis holds for some arbitrary fixed height $h_0 \geq 1$.
	For $h = h_0 + 1$ we distinguish two different attack targets of the adversary.
	
	If the attack takes place outside $T$, then $T$ remains connected and analogous to above $\ap$ is connected to all survivors via $a'$.
	
	If the adversary attacks a player inside $T$, then $\ap$ is connected to $a$ via $p(a)$ which is equivalent to $\ap$ having a direct edge to $a$ and thus, by induction hypothesis, no additional edge should be bought into any subtree of $a$.
\end{proof}
The next statements establish that we can treat all subtrees of a node of $M$ independently, that buying one edge into a subtree suffices if the edge-buy decision was negative for all subtrees one level below and that \textsc{RootedMetaTreeSelect} is correct.
\begin{lemma} \label{lm_treeindepence}
	Assume the existence of an edge from $\ap$ to some node $a$ of $M$ and let $V_1, \dots, V_z$ denote the node sets of $a$'s subtrees.
	
	If there are optimal partner sets $O_1,\dots,O_z$ for $C$ which all contain $r^*$, then there is an optimal partner set $O^*$ for $C$ containing $r^*$ where $O^* \cap V_j = O_j \cap V_j$ for all $1\leq j \leq z$.
\end{lemma}
\begin{proof}
	The choice of nodes to connect to within any subtree $T$ of $a$ does not influence $\ap$'s connectivity to nodes outside of $T$ as every path from some node in $T$ to some node outside $T$ traverses node $a$.
	Thus, the choice of nodes in $T$ for all optimal partner sets for $C$ is independent from the choice of nodes in other subtrees.
\end{proof}
\begin{lemma}
	\label{cl:OneEdgeIsEnough}
	Consider any subtree $T$ with root $a$ of $M$.
	If \textsc{RootedMetaTreeSelect} has correctly decided that no edge should be bought into any subtree of $a$, then every optimal partner set for component $C$ which contains $r^*$ has at most one node from within $T$.
\end{lemma}

\begin{proof}
	By induction over height $h$ of $T$.
	For $h=0$ the block $a$ must be a leaf of the \metatree.
	By Lemma~\ref{lm:mtLeafIsCB} and Lemma~\ref{lm:oneEdgeToCB} the statement follows.
	
	Assume now, that the hypothesis holds for some height $h_0\geq 0$. For $h = h_0 + 1$ we assume the existence of an edge between $\ap$ and $p(a)$ and distinguish two different attack targets of the adversary:
	
	If the attack takes place outside of $T$, then $T$ remains connected and thus at most one edge into $T$ suffices.
	
	If a node within $T$ is attacked, then $\ap$ is connected to $a$ via $p(a)$ which is equivalent to $\ap$ having a direct edge to $a$, which is the  assumption under which \textsc{RootedMetaTreeSelect} has correctly decided not to buy edges into any subtree of $a$.
\end{proof}
\begin{lemma} \label{lm:rootedmetatreecorrect}
	If $\ap$ has an edge to $p(r_T)$ and $opt(r_T)$ is returned by \textsc{RootedMe\-ta\-Tree\-Select($M(r),r_T$)}, then there exists an optimal partner set for component $C$ which contains $r^*$ and $opt(r_T)$.
\end{lemma}
\begin{proof}
	\textsc{RootedMetaTreeSelect} recurses on all subtrees of $r_T$.
	
	If $r_T$ is a leaf of $M(r)$, then the algorithm computes the expected profit contribution of buying a single edge into the leaf and if this exceeds $\alpha$ an immunized node within this leaf is eventually added to $opt(r_T)$.
	Since, by Lemma~\ref{lm:mtLeafIsCB}, leaves are \candidateblocks, buying at most one edge is correct and the algorithm compares both possible options. Otherwise, if $r_T$ is not a leaf, then there are three cases:
	
	\textbf{Case 1:} If $r_T$ is a \bridgeblock, then $p(r_T)$ must be a \candidateblock and by assumption $\ap$ is connected to $r_T$ via $p(r_T)$.
	Thus, by Lemma~\ref{lm_treeindepence}, there is an optimal partner set for $C$ which contains $r^*$ and the union of all the returned node sets obtained by recursing on all children of $r_T$.
	
	\textbf{Case 2:} If there exists an edge between $\ap$ and any node in the subtree rooted at $r_T$.
	Then, by Lemma~\ref{lm:noadditionaledges} and Lemma~\ref{lm_treeindepence}, there is an optimal partner set for $C$ which contains $r^*$ and the union of all the returned node sets obtained by recursing on all children of $r_T$.
	
	\textbf{Case 3:} If $opt(r_T) = \emptyset$ after all children of $r_T$ were processed.
	Then, by Lemma~\ref{cl:OneEdgeIsEnough}, it is optimal to consider buying at most one edge into the subtree rooted at $r_T$.
	\textsc{RootedMetaTree\-Select} tries all possibilities and selects the most profitable one.
\end{proof}
\noindent Finally, we show that \textsc{MetaTreeSelect} is correct.
\begin{restatable}{theorem}{restateMTCorrectTheorem}\label{th:MTCorrect}
	If there is an optimal partner set with at least two nodes for component $C$, then \textsc{MetaTreeSelect} algorithm outputs such a set.
\end{restatable}
\begin{proof}
	Assume that there exists an optimal partner set with at least two nodes for component $C$ and assume that the \metatree $M$ of component $C$ is rooted at some leaf $r$.
	Since the algorithm compares all possibilities to root $M$ at a leaf and by Lemma~\ref{lm:mt-connect-to-leaves}, at least one of those leaves must be contained in an optimal partner set.
	Assume that $r$ is indeed such a leaf.
	
	Thus, by buying $r$ we satisfy the assumption needed for \textsc{RootedMetaTreeSelect}.
	By Lemma~\ref{lm:rootedmetatreecorrect}, \textsc{RootedMetaTreeSelect} returns a set of nodes, which together with $r^*$ yields an optimal partner set for $C$.
	Hence, the algorithm \textsc{MetaTreeSelect} is correct.
\end{proof}

\subsection{Cost Analysis}\label{sec:costanalysis}
We now prove that our algorithm is efficient.
\begin{theorem}\label{thm:cost}
    \textsc{BestResponseComputation} runs in time $\mathcal{O}(n^4 + k^5)$, where $n$ is the number of nodes in the network and $k$ is the number of blocks in the largest occurring \metatree. Only in terms of $n$, the run time is in $\mathcal{O}(n^5)$.
\end{theorem}
\begin{proof}[of Theorem~\ref{thm:cost}]
    We start with analyzing the five core subroutines of our algorithm:
    \textsc{SubsetSelect}, \textsc{Greedy\-Select}, \textsc{PartnerSetSelect}, \textsc{MetaTreeConstruct} and \textsc{MetaTreeSelect}.
    
    \textsc{SubsetSelect}: This subroutine calculates a three dimensional matrix where every dimension has a maximum length of $n$.
    The calculations needed for one matrix entry can be done in constant time.
    This results in a time complexity of $\mathcal{O}(n^3)$.
    
    \textsc{GreedySelect}: Every single component of $\CU$ is tested individually if an edge into it is beneficial.
    For the calculation of the expected profit contribution only the size of the component is needed. Thus, this task can be done for all components in $\mathcal{O}(n)$.
    
    \textsc{PartnerSetSelect}: This subroutine is executed for every component $C\in \CI$. Let $p$ denote the number of nodes of any processed component $C\in \CI$ and let $q$ denote its number of edges. The algorithm covers three cases:
   
    Case 1: No edge is bought into $C$:
    Here an attack is simulated on every targeted node $y$ in component $C$ and the remaining profit is calculated.
    Calculating the profit after an attack on node $y$ can be done in linear time $\mathcal{O}(p+q)$ using breadth first search (BFS). Since there are $\mathcal{O}(p)$ many targeted nodes in $C$ the overall cost is in $\mathcal{O}(p(p+q))$. In the worst case this is in $\mathcal{O}(p^3)$.

    Case 2: One edge is bought into $C$:
    The profit of every combination of buying an edge to node $x$ in $C$ and having node $y$ in $C$ attacked is calculated.
    This amounts to a total cost of $\mathcal{O}(p^2(p+q))$, which is in $\mathcal{O}(p^4)$.

    Case 3: Two or more edges are bought into $C$:
    In this case \textsc{MetaTreeConstruct} and \textsc{MetaTreeSelect} are executed.
    \begin{itemize}
    \item \textsc{MetaTreeConstruct}: 
    To construct the \metatree we first create the \metagraph of $C$.
    To do so we use BFS from an uncollapsed node to find all adjacent nodes of the same types.
    All of these nodes are then merged and disregarded in further iterations of BFS. There may be $\mathcal{O}(p)$ calls to BFS and (assuming a representation as adjacency list) merging neighboring nodes of the same type can be done in $\mathcal{O}(p+q)$.
    Thus the construction of the \metagraph happens in $\mathcal{O}(p(p+q))$ which is in $\mathcal{O}(p^3)$.

    The construction of the \metatree consists of two steps:
    Finding cycles to collapse and collapsing the nodes on the cycle.
    The first part can be done using BFS and in every iteration of BFS at least one node is collapsed.
    Thus, the overall cost for all BFS computations is in $\mathcal{O}(p(p+q))$.
    The second part can be done by checking the edges of all neighboring nodes.
    This can be done in $\mathcal{O}(p+q)$ time per merge. 
    Hence we have a total cost of $\mathcal{O}(p(p+q))$ which is in $\mathcal{O}(p^3)$ for the construction of the \metatree.
    
    \item \textsc{MetaTreeSelect}: Let $k$ denote the number of nodes of the created \metatree.
    The \textsc{MetaTreeSelect} algorithm calculates the expected profit contribution for connecting with the currently considered component via two or more edges.
    It roots the \metatree at every leaf, and then traverses it bottom up.
    During this traversal the following calculation is made at every node:
    The number of reachable nodes from a node $x$ given an attack on node $y$ is evaluated for every possible node pair $x,y$ of the \metatree.
    As seen above this cannot exceed a run time of $\mathcal{O}(k^3)$.
    As this is done for every node for every possible root, the total cost is in $\mathcal{O}(k^5)$.
    \end{itemize}
    Thus Case 3 has total cost in $\mathcal{O}(p^3+k^5)$. It follows that the total cost of \textsc{PartnerSetSelect} can be upper bounded by $\mathcal{O}(p^4+k^5)$. 
    
    \textsc{Possible\-Strategy:} First, an arbitrary node for every previously selected vulnerable component is chosen. This can be done in constant time per component. Adding the corresponding edges and updating immunized, vulnerable and targeted regions can be done via at most $n$ BFS computations and thus in $\mathcal{O}(n^3)$. Next, \textsc{PartnerSetSelect} is called for each of the $c$ components in $\CI$. Let $p_j$ and $k_j$ denote the number of nodes and the size of the \metatree for component $C_j \in CI$ for $1\leq j\leq c$. 
    Therefore, the cost of \textsc{PossibleStrategy} can be upper bounded by $\mathcal{O}(n^3 + \sum_{1\leq j \leq c}(p_j^4 + k_j^5))$ which is in $\mathcal{O}(n^4 + k^5)$.    
    
    \textsc{BestResponseComputation:} The total cost for \textsc{BestResponseComputation} thus is in $\mathcal{O}(n^3 + n + n^3 + n^4 + k^5)= \mathcal{O}(n^4 + k^5)$. Since $k<n$ this yields cost in $\mathcal{O}(n^5)$.
\end{proof}

%% file: EmpiricalResults.tex

\subsection{Empirical Results}
\label{sec:empirical}

To evaluate our proposed best response algorithm, we apply it to randomly generated networks.
To allow a direct comparison with the extensive experiments by Goyal et al.~\cite{GJKKM15arxiv,GJKKM15}, we chose the same setup, i.e. the initial networks were generated via the Erd\H{o}s-Renyi model with average degree $5$ and $\alpha=\beta=2$. 

In Fig.~\ref{fig:emp-plot0}~(left) we show the number of rounds required until the best response dynamic arrives at a Nash equilibrium averaged over 100 experiments per configuration.
\begin{figure}[htb]
  \centering
  \includegraphics[width=5.0cm]{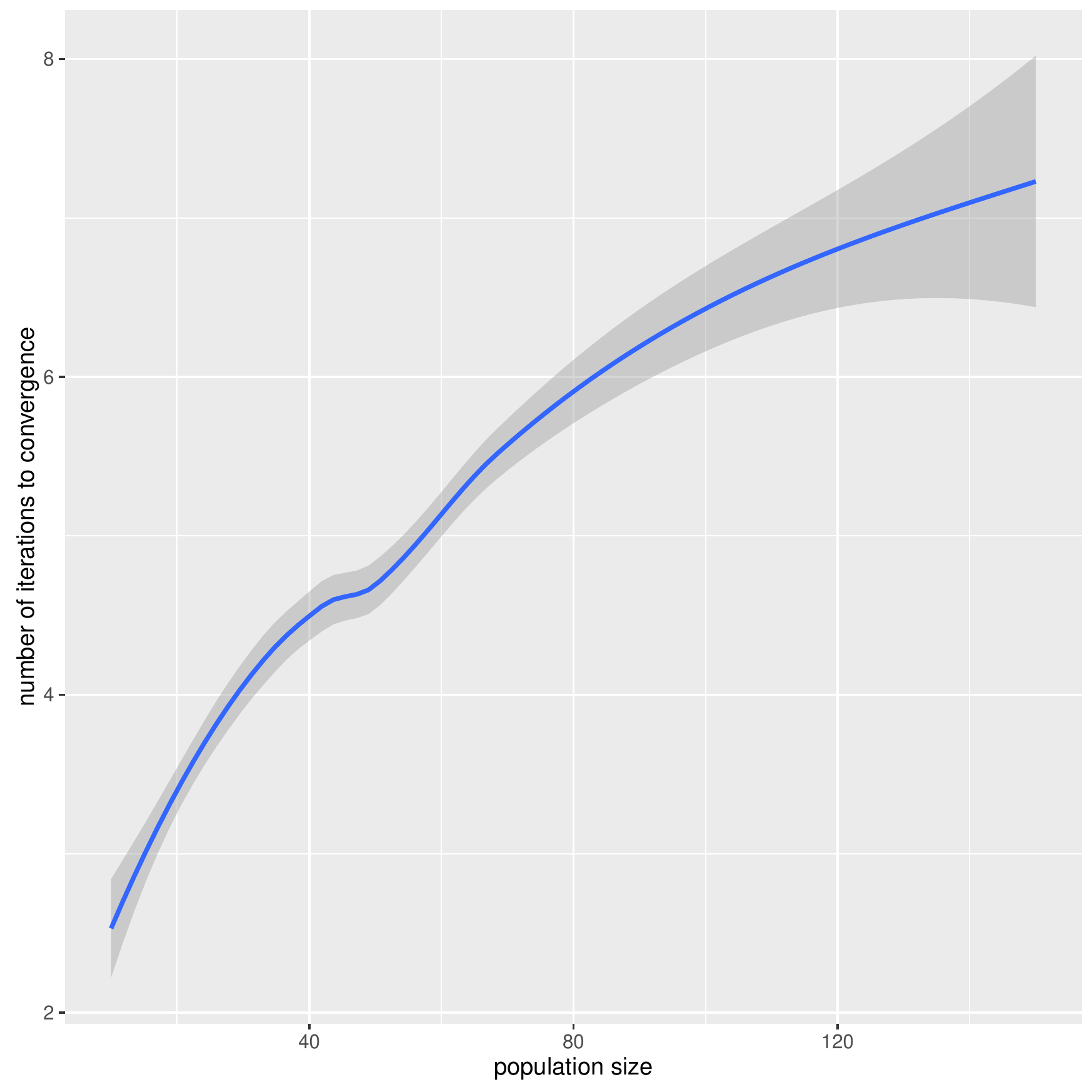}
  \hfill
  \includegraphics[width=5.0cm]{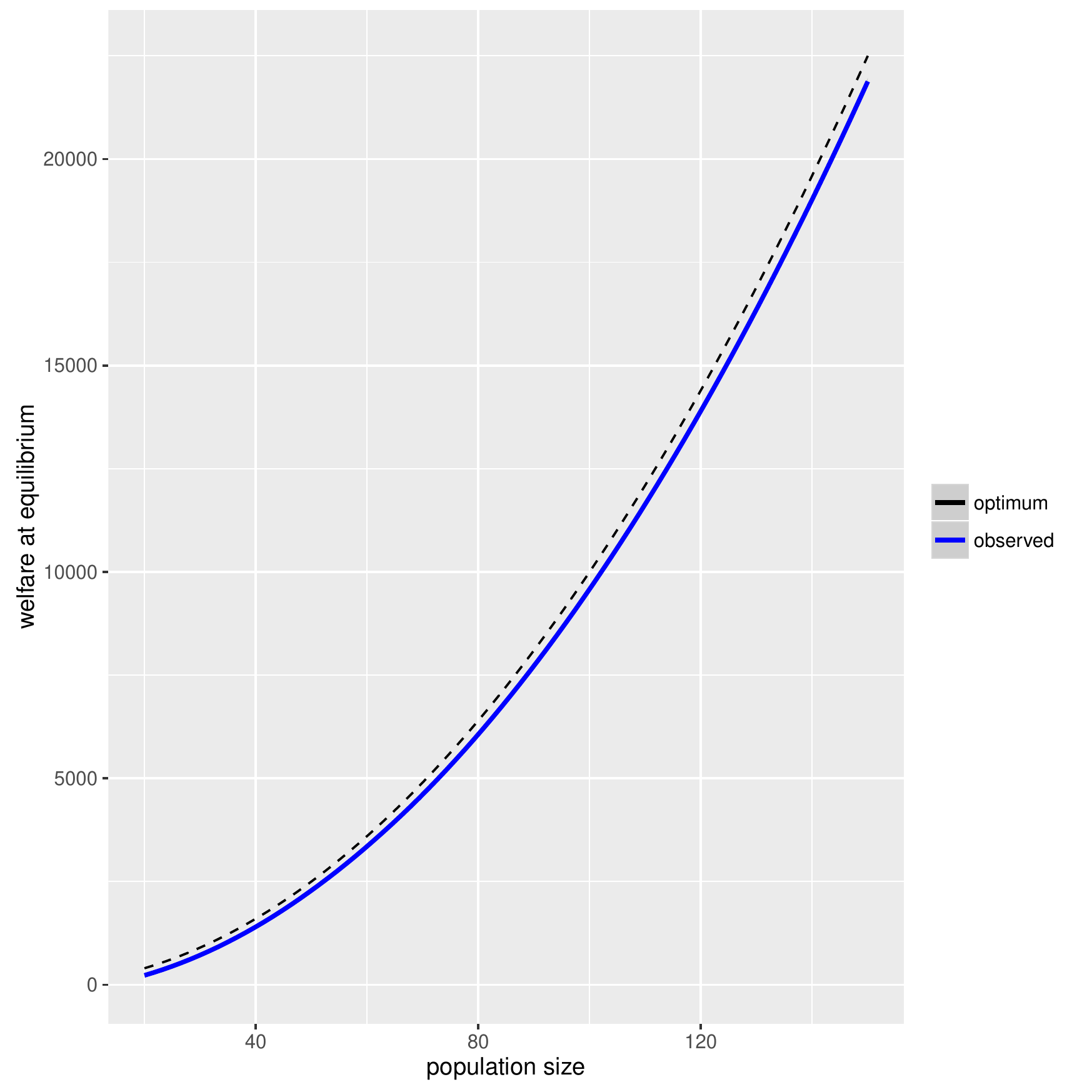}
  \hfill
  \includegraphics[width=5.0cm]{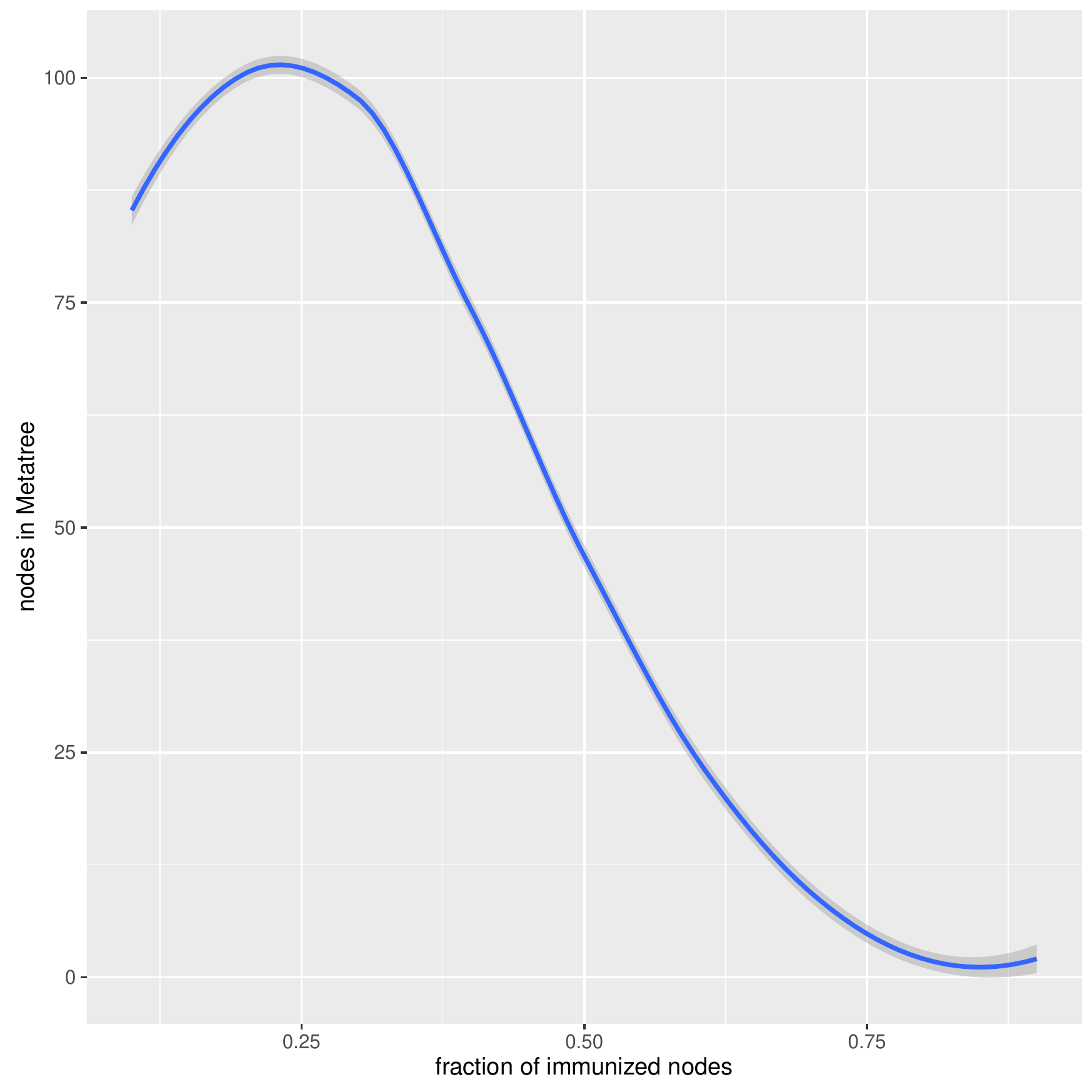}
  \caption{Number of iterations until convergence (top);  average welfare at non-trivial equilibria (middle) $\alpha = \beta = 2$.
    Bottom: \metatree size with respect to immunization density.}
  \label{fig:emp-plot0}
\end{figure}
Here a round consists of a best response strategy update by every player in some fixed order.
In contrast with the results in~\cite{GJKKM15arxiv,GJKKM15}, where much weaker strategy updates called swapstable best response have been used, our experiments indicate a speed-up of 50\% by updating to best possible strategies.

Fig.~\ref{fig:emp-plot0} (middle) plots the welfare of networks at
(non-trivial) equilibria over the population size achieved by performing best response dynamics with the initial setting as described above.
For every configuration a random sample from the 100 independent experiments was chosen.
We observe that the achieved global welfare is quite close to the optimal value of $n(n-\alpha)$ which indicates that best response dynamics yield favorable outcomes whenever they converge to a Nash equilibrium\footnote{Note that this property is not guaranteed, since Goyal et al.~\cite{GJKKM15} present a best response cycle, that is, a configuration where a sequence of best response strategy updates returns to the initial network and may thus never converge.}.


The run time for computing the partial best response strategy for attaching to connected components from $\CI$, as described in Section~\ref{subsec:mt-Select}, heavily depends on the number of \candidateblocks in the resulting \metatree for the component.   
To illustrate the benefits of using the \metatree as a preprocessing step, we compared the number of nodes of a random connected network and the number of \candidateblocks in the corresponding \metatree, varying the fraction of immunized players in the initial network. For each experiment we used connected $G_{n,m}$ random networks with $n=1000$ and $m=2n$ edges. 

The results are depicted in Fig.~\ref{fig:emp-plot0}~(right). There the number of \candidateblocks of the \metatree over
the fraction of immunized players averaged over $100$ runs per parameter combination is shown. As
expected, the number of \candidateblocks in the \metatree shrinks rapidly as the fraction
of immunized vertices in the initial network increases. We also note that the maximum number of \candidateblocks
in the \metatree is roughly $10\%$ of the number of nodes in the initial network.

We illustrate the behavior of best response dynamics with
Fig.~\ref{fig:emp-samplerun}, where snapshots from a sample run are shown. The initial network is sparsely connected
with $n/2 = 25$ edges and contains no immunized players.  During the
first round, a player with a sufficient number of incident edges bought by other players
decides to obtain immunization. Then all following players choose to
connect to this newly immunized hub. The subsequent rounds distribute
players away from the newly formed targeted regions until an
equilibrium is achieved after four rounds.

\begin{figure}[htb]
  \centering
  \includegraphics[width=2.5cm]{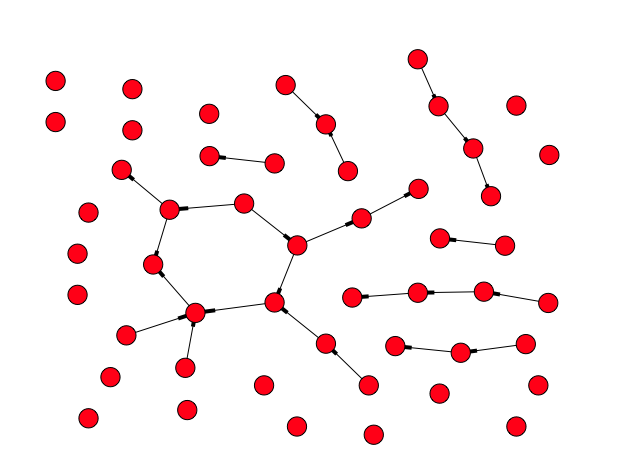}
  \hfill
  \includegraphics[width=2.7cm]{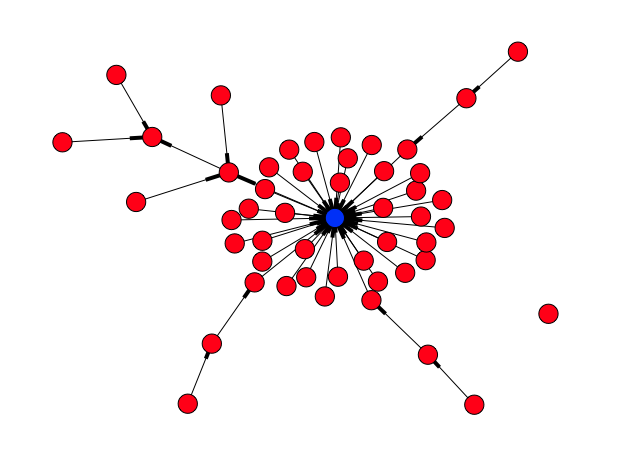}
  \hfill
  \includegraphics[width=2.7cm]{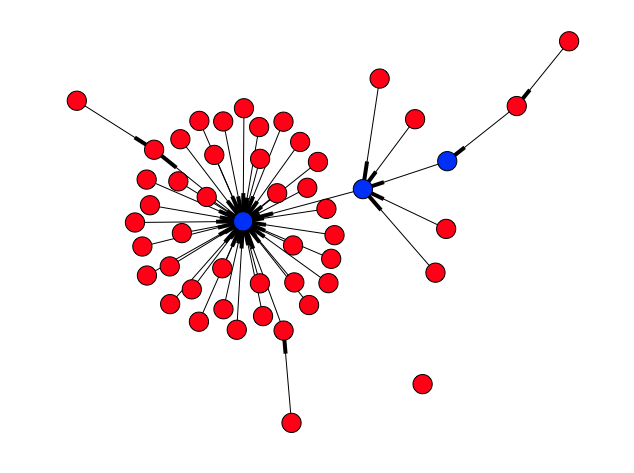}
  \hfill
  \includegraphics[width=2.7cm]{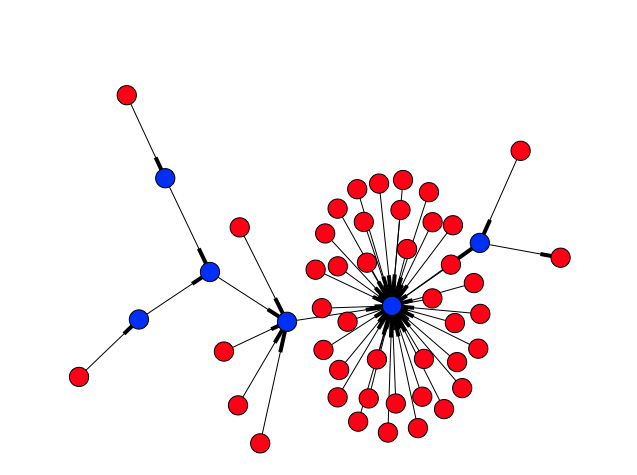}
  \hfill
  \includegraphics[width=2.7cm]{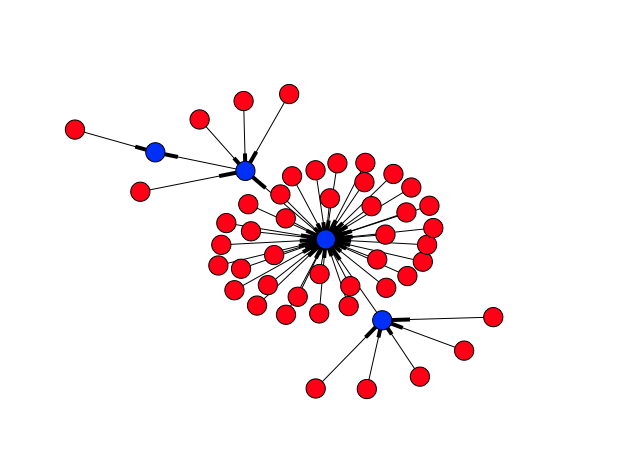}
  \caption{A sample run of the best response dynamic ($n=50, \alpha = \beta = 2$). \hspace{\textwidth} Initial state. \hfill After round 1. \hfill After round 2. \hfill After round 3. \hfill After round 4.}
  \label{fig:emp-samplerun}
\end{figure}

%% file: Uniform-attacker.tex

\section{Random Attack Adversary}
\label{sec:uniformAttacker}
In the following we will consider the random attack adversary~\cite{GJKKM15} which is a less predictable adversary compared to the maximum carnage adversary.

\paragraph{The Modified Model and its Implications}
The random attack adversary attacks one vulnerable node uniformly at random, i.e. with probability $\frac{1}{|\mathcal{U}|}$ per player in $\mathcal{U}$.
Facing this adversary, the probability for a vulnerable player $v_\mathcal{U}$ to be destroyed increases linearly with the size of player $v_\mathcal{U}$'s vulnerable region $\mathcal{R}_{\mathcal U}(v_\mathcal{U})$.
Moreover, every vulnerable player belongs to the set of targeted vertices, that is, $\mathcal{T} = \mathcal{U}$.

Clearly, the behavior of the chosen adversary influences the utility of a player's strategy heavily.
Thus, it is not at all clear that our algorithm for computing a utility maximizing strategy for the maximum carnage adversary can be adapted to the random attack adversary.
We now show that this can actually be done with only some minor adjustments.
Thereby we increase the run time slightly, but our algorithm is still efficient.

\paragraph{Adaptation of our Algorithm}
With the deterministic maximum carnage adversary, a player was either risk-free or potentially targeted with probability $\frac{1}{|\mathcal{T}|}$.
Introducing the random attack adversary results in up to $n$ different probabilities for the player $\ap$ to be destroyed due to different possible sizes of targeted regions.

Note that none of the subroutines \textsc{GreedySelect}, \textsc{MetaTreeConstruct} and \textsc{MetaTreeSelect} depend on $t_{max}$, but are only formulated in terms of $\mathcal{T}$ and $\mathcal{R}_{\mathcal T}$.
Even though these sets change in comparison to the maximum carnage adversary, correctness still holds for those parts of the algorithm.
In particular, the \metatree still fulfills the same properties as before, even though the number of \bridgeblocks increases for many input graphs, as all vulnerable regions are targeted regions now (see Fig.~\ref{fig:sample-metas-uniform}).

Thus, only the subroutine \textsc{SubsetSelect} has to be adjusted.

\begin{figure}[htb]
	\centering
	\includegraphics[width=4cm]{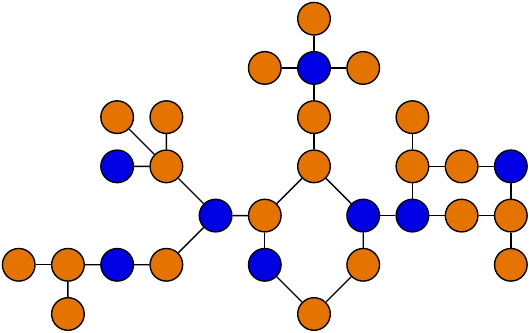}
	\hfill
	\includegraphics[width=4cm]{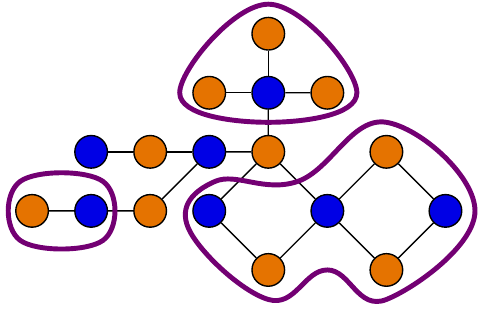}
	\hfill
	\includegraphics[width=3cm]{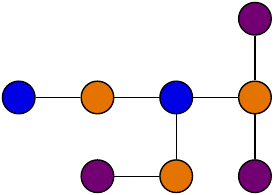}
	
	~
	
	\hfill
	\includegraphics{vulnerable.pdf} vulnerable
	\hfill
	\includegraphics{targeted.pdf} targeted
	\hfill
	\includegraphics{mixed.pdf} mixed
	\hfill
	\includegraphics{immunized.pdf} immunized
	\hfill
	\caption{A graph component (left) and its \metagraph (middle) and \metatree (right) for the random attack adversary.}
	\label{fig:sample-metas-uniform}
\end{figure}

\paragraph{\textsc{SubsetSelect} for the Random Attack Adversary}
With the random attack adversary the maximum expected profit generated by all possible sizes of the vulnerable region of $\ap$ need to be considered.
A larger region might lower the overall profit as it increases the attack probability on the vulnerable region of $\ap$.
Instead of two sets of selected components $\CSet_v, \CSet_t$, we now need to consider up to $n$ sets that maximize the profit for each possible size of $\mathcal{R}_\mathcal{U}$.

As before, the interdependent subset selection on $\CU \setminus \Cinc$ with $|\CU \setminus \Cinc| = m$ fills the matrix $M$.
In the end, the plane $M[m,\cdot,\cdot]$ contains all possible numbers of players from these components that $\ap$ can buy edges to.
These values constitute all possible sizes of $\mathcal{R}_\mathcal{U}(v)$.
There might be different subsets of $\CU \setminus \Cinc$ that result in the same size of the active player's vulnerable region.
Observe that the maximum utility is always achieved with the subset that uses the least amount of edges.

So for every row $M[m,\cdot,z]$ the algorithm selects the minimum $y$ so that $M[m,y,z] = z$.
This can be done in $\mathcal{O}(n^2)$ and yields a maximum of $n$ values.
We call this modified version \textsc{UniformSubsetSelect}. 

\paragraph{\textsc{BestResponseComputation} for the Random Attack Adversary}
The modified main algorithm can be found in Algorithm~\ref{alg:bestResponseUniform}. 

The only change compared to the maximum carnage adversary is that for each of the $\mathcal{O}(n)$ solutions of \textsc{UniformSubsetSelect} the subroutine \textsc{PossibleStrategy} is executed.
This yields an overall run time of $\mathcal{O}(n^3 + n(n^4+k^5)) = \mathcal{O}(n^5 + nk^5)$, which is in $\mathcal{O}(n^6)$.

Observe that we calculate best responses for all possible sizes of $\mathcal{R_U}(\ap)$, and $\mathcal{T}$ remains unchanged across all cases.
\begin{algorithm}
    \KwIn{Strategies $\mathbf{s} = (s_1, \ldots, s_n)$, active player $\ap, 1 \leq \aps \leq n$}
    \KwOut{Best response strategy of player $\ap$ denoted by $s_\aps = (x_\aps, y_\aps)$}

    $s_\emptyset = (\emptyset, 0)$\;
    
    Let $G(\mathbf{s'})$ be the induced game state with $\mathbf{s'} = (s_1, \ldots, s_{\aps-1}, s_\emptyset, s_{\aps+1}, \ldots, s_n)$\;

    Let $\mathfrak{A}$ be the solutions of \textsc{UniformSubsetSelect} on $\CU$\;

    Let $\CSet_g$ be the solution of \textsc{GreedySelect} on $\CU$\;

    $S_v$ = $\left\{\textnormal{\textsc{PossibleStrategy}}(\CSet, 0) \mid \CSet \in \mathfrak{A} \right\}$\;

    $s_g$ = \textsc{PossibleStrategy}($\CSet_g$, 1)\;

    $S = \{s_{\emptyset}, s_g\} \cup S_v$\;
    \Return{\textnormal{strategy $s \in S$ which maximizes $\ap$'s utility}}\;

    \caption{\textsc{BestResponseComputation} for the random attack adversary.}
    \label{alg:bestResponseUniform}
\end{algorithm}

%% file: Conclusion.tex
\section{Conclusion}
\label{sec:conclusion}

For most models of strategic network formation computing a utility maximizing strategy is known to be \NP-hard.
In this paper, we have proven that the model by Goyal et al.~\cite{GJKKM15,GJKKM15arxiv} is a notable exception to this rule.
The presented efficient algorithm for computing a best response for a player circumvents a combinatorial explosion essentially by simplifying the given network and thereby making it amenable to a dynamic programming approach.
An efficient best response computation is the key ingredient for using the model in large scale simulations and for analyzing real world networks. Moreover, our algorithm can be adapted to a significantly stronger adversary and we are confident that further modifications for coping with other variants of the model are possible.

\paragraph{Future Work}

Settling the complexity of computing a best response strategy with respect to the maximum disruption adversary is left as an open problem. 
Besides this, it seems worthwhile to consider a variant with directed edges, originally introduced by Bala \& Goyal~\cite{BG00}.
Directed edges would more accurately model the differences in risk and benefit which depend on the flow direction.
Using the analogy of the WWW, a user who downloads information benefits from it, but also risks getting infected.
In contrast, the user providing the information is exposed to little or no risk.
Moreover, a constant cost for immunization seems unrealistic.
In reality a highly connected node would have to invest much more into security measures than any node with only a few connections.
Thus, it would be interesting to consider a version where immunization costs scale with the degree of a node.
We believe that this variant of the model yields more diverse optimal networks and also a greater variety of equilibria.